\documentclass[a4paper,12pt]{article}
\usepackage{latexsym,amssymb,amsfonts,amsmath,amsthm}
\usepackage[dvips]{graphicx}
\usepackage{comment}

\newtheorem{theorem}{Theorem}
\newtheorem{lemma}{Lemma}
\newtheorem{corollary}{Corollary}
\newtheorem{proposition}{Proposition}
\newtheorem{remark}{Remark}

\numberwithin{theorem}{section}
\numberwithin{lemma}{section}
\numberwithin{corollary}{section}
\numberwithin{proposition}{section}
\numberwithin{remark}{section}

\newcommand{\supp}{{\rm supp}}

\usepackage{color}

\title{{\Large {\bf Explicit expression of scattering operator of some quantum walks on impurities }}
\author{
{\small Takashi Komatsu\footnote{Email: komatsu@coi.t.u-tokyo.ac.jp}}\\
{\scriptsize Department of Bioengineering School of Engineering,The University of Tokyo, }\\  
{\scriptsize Bunkyo, Tokyo, 113-8656, Japan. }\\
{\small Norio Konno\footnote{Email: konno-norio-bt@ynu.ac.jp }}\\
{\scriptsize Department of Applied Mathematics, Yokohama National University, }\\  
{\scriptsize Hodogaya, Yokohama 240-8501, Japan.}\\
{\small Hisashi Morioka\footnote{Email: morioka@cs.ehime-u.ac.jp}}\\
{\scriptsize Graduate School of Science and Engineering, Ehime University, }\\  
{\scriptsize Bunkyo-cho 3, Matsuyama, Ehime, 790-8577, Japan.}\\
{\small Etsuo Segawa\footnote{Email: segawa-etsuo-tb@ynu.ac.jp}}\\
{\scriptsize Graduate School of Education Center and Graduate School of Environment Information Sciences, Yokohama National University, } \\
{\scriptsize Hodogaya, Yokohama 240-8501, Japan.}}
}
\date{\empty }
\begin{document}
\maketitle

\par\noindent
\begin{small}
\par\noindent
{\bf Abstract}. 
In this paper, we consider the scattering theory for a one-dimensional quantum walk with impurities which make reflections and transmissions.
We focus on an explicit expression of the scattering operator.
Our construction of the formula is based on the counting paths of quantum walkers. 
The Fourier transform of the scattering operator gives an explicit formula of the scattering matrix which is deeply related with the resonant-tunneling for quantum walks.
%\footnote[0]{
%{\it Abbr. title:}Construction theorem of the positive support of $n$-th power of the Grover walk
%}
%\footnote[0]{
%{\it AMS 2000 subject classifications: }
%60F05, 05C50, 15A15, 05C60
%%}
%\footnote[0]{
%{\it PACS: } 
%03.67.Lx, 05.40.Fb, 02.50.Cw
%}
\end{small}

\section{Introduction}

Quantum walks have been studied in both finite and infinite systems. 
For finite systems, for example, studies on the effectiveness and universality of quantum walks in the quantum search algorithms have been studied, see \cite{Ambainis2003}, \cite{Childs} \cite{Portugalbook} and its references therein.
On the other hand, for an infinite system, there are several mathematical works obtaining different limiting behavior from classical random walks~\cite{Konno2008b}.
In particular, recently, studies on quantum walks in view of the spectral theory and the scattering theory have been intensively studied, for examples, \cite{FelHil1}, \cite{FelHil2}, \cite{MS4},  \cite{Morioka}, \cite{MoSe}, \cite{Su}. 
In this paper, we also consider one-dimensional position-dependent quantum walks (QW for short) in view of the scattering theory.
As will been mentioned later, the method of these works are based on the scattering theory of quantum mechanics like Schr\"{o}dinger equations.
For general information of this research area, the monograph by Yafaev \cite{Yafaev} is available and its reference is also worthwhile. 
On the other hand, our method is simply  based on a kind of combinatorial approach using the primitive form of QW by Feynmann and Hibbs (1965) \cite{FH}.
%We consider quantum walk on one-dimensional lattice. 

Now let us introduce the model.
The total Hilbert space is denoted by $\mathcal{H}:=\ell^2(\mathbb{Z};\mathbb{C}^2)\cong \ell^2(A)$. 
Here $A$ is the set of arcs of one-dimensional lattice whose elements are labeled by $\{(x;R),(x;L) \;|\; x\in \mathbb{Z}\}$, 
where $(x;R)$ and $(x;L)$ represents the arcs ``from $x-1$ to $x$", and ``from $x+1$ to $x$", respectively. 
We assign a $2\times 2$ unitary matrix to each $x\in \mathbb{Z}$ so called local quantum coin 
	\[ C_x=\begin{bmatrix} a_x & b_x \\ c_x & d_x \end{bmatrix}. \]
Putting $|L\rangle := [1,0]^\top$, $|R\rangle:=[0,1]^\top$ and $\langle L|=[1,0]$, $\langle R|=[0,1]$, 
we define the following matrix valued weights associated with moving to left and right from $x$ by 
	\[ P_x=|L\rangle \langle L| C_x,\;\;Q_x=|R\rangle \langle R| C_x, \]
respectively. Then the time evolution operator on $\ell^2(\mathbb{Z};\mathbb{C}^2)$ is described by 
	\[ (U\psi)(x)=P_{x+1}\psi(x+1)+Q_{x-1}\psi(x-1) \]
for any $\psi\in \ell^2(\mathbb{Z};\mathbb{C}^2)$. 
Its equivalent expression on $\ell^2(A)$ is described by 
	\begin{align}\label{def:U'} 
        (U'\phi)(x;L) &= a_{x+1}\phi(x+1;L)+b_{x+1}\phi(x+1;R), \notag \\
        (U'\phi)(x;R) &= c_{x-1}\phi(x-1;L)+d_{x-1}\phi(x-1;R) 
        \end{align}
for any $\psi\in \ell^2(A)$. 
We call $a_x$ and $d_x$ the transmitting amplitudes, and 
$b_x$ and $c_x$ the reflection amplitudes at $x$, respectively
\footnote{If we put $a_x=d_x=1$ and $b_x=c_x=\sqrt{-1}=i$, then the primitive form of QW in \cite{FH} is reproduced.}.
Remark that $U$ and $U'$ are unitarily equivalent such that 
letting $\eta:\ell^2(\mathbb{Z};\mathbb{C}^2)\to \ell^2(A)$ be 
	\begin{align*}
        (\eta \psi)(x;R)= \langle R|\psi \rangle,\;\;(\eta \psi)(x;L)= \langle L|\psi \rangle
	\end{align*} 
then we have $U=\eta^{-1}U'\eta$. 
The free quantum walk is the quantum walk where all local quantum coins are described by the identity matrix i.e.
$$
(U_0 \psi )(x)= \left[ \begin{array}{cc} 1 & 0 \\ 0 & 0 \end{array} \right] \psi (x+1) + \left[ \begin{array}{cc} 0 & 0 \\ 0 & 1 \end{array} \right] \psi (x-1).
$$
Then the walker runs through one-dimensional lattices without any reflections in the free case. 

In this paper we set ``impurities" on 
\[ \Gamma_M:=\{0,1,\dots,M-1\} \] 
in the free quantum walk on one-dimensional lattice; that is, 
	\begin{equation}\label{inpurity} 
        C_x=\begin{cases} \begin{bmatrix}a & b \\ c & d\end{bmatrix} & \text{: $x\in \Gamma_M$,}\\ \\ I_2 & \text{: $x\notin \Gamma_M$.} \end{cases} 
        \end{equation}
In view of the scattering theory, the operator $U$ is a finite rank perturbation of $U_0$.
Quantum walkers move without reflections by $U_0$.
For quantum walkers in the time evolution by $U$, reflections and transmissions occur on $\Gamma_M$ due to the matrix $C_x $.
Then we can deal $U$ and $U_0$ with an analogue of the scattering of the one-dimensional Schr\"{o}dinger equation.

%We define the unitary operator $U$ on $\mathcal{H} = \ell^2 (\mathbb{Z} ; \mathbb{C}^2 )$ by
%$$
%(U\psi )(x)= C_1 (x+1) \psi (x+1)+ C_2 (x-1) \psi (x-1), \quad x\in \mathbb{Z} ,
%$$
%for $ \psi \in \mathcal{H} $ where $C(x):= C_1 (x )+C_2 (x)\in U (2)$ for every $x\in \mathbb{Z}$.
%Letting $C$ be the operator of multiplication by the matrix $C(x)$, we have $U=SC$ where $S$ is the shift operator defined by
%$$
%(S\psi )(x)= \left[ \begin{array}{c}
%\psi _{\leftarrow} (x+1) \\ \psi _{\rightarrow} (x-1) \end{array} \right] , \quad \psi \in \mathcal{H} .
%$$
%The operator $U$ is the operator of time-evolution of two states QWs on $\mathbb{Z}$.
%Taking an initial state $ \psi_0 \in \mathcal{H}$, the state of quantum walker at time $t\in \mathbb{Z} $ is given by $ \psi (t,\cdot )=U^t \psi_0 $.

%In order to derive the scattering theory for $U$, we also consider the operator for the free quantum walk.
%Let $ U_0 =SC_0$ where $C_0$ is the operator of multiplication by the $2\times 2$ identity matrix $I_2$ at each $x\in \mathbb{Z}$.
%Thus $U_0$ is also a unitary operator on $\mathcal{H}$.
%In this paper, we assume that
%\begin{equation}
%C(x)=I_2 \quad \text{except for a finite number of } x\in \mathbb{Z}.
%\label{S1_assumption}
%\end{equation}

Roughly speaking, there are two manners of the scattering theory.
One is the time-dependent theory.
In the time-dependent theory, the wave operator is a fundamental subject.
The wave operator for QWs is given by
\begin{equation}
W_{\pm} = {\mathop{\text{s-lim}}_{t \to \pm \infty}} \, U^{-t} U_0^t \quad \text{in} \quad \mathcal{H} .
\label{S1_def_waveop}
\end{equation}
The wave operator satisfies the following property.
For the proof, see Suzuki \cite{Su}.

\begin{theorem}
The wave operator exists and are complete i.e. the range of $W_{\pm}$ is the absolutely continuous subspace $\mathcal{H}_{ac} (U)$ for $U$.
Precisely, for any $\phi \in \mathcal{H}_{ac} (U)$, there exist $\psi_{\pm} \in \mathcal{H}$ such that $\| U_0^t \psi_{\pm} -U^t \phi \| _{\mathcal{H}} \to 0$ as $ t\to \pm \infty $.
The wave operator is unitary on $\mathcal{H}$ and we have $W _{\pm} ^* = W_{\mp} ^{-1} $.
\label{S1_thm_suzuki}
\end{theorem}

Another manner of the scattering theory is the time-independent theory.
In this manner, we study generalized eigenfunctions of $U$ and $U_0$.
Here generalized eigenfunctions belong to functional spaces larger than $\mathcal{H}$.
For the scattering theory, we consider them in the Banach space $ \ell^{\infty} (\mathbb{Z} ; \mathbb{C}^2)$.
In the generalized eigenfunction of $U$, the \textit{scattering matrix} appears as the amplitudes of the reflected wave and the transmitted wave.
The scattering matrix is given by the spectral decomposition of the \textit{scattering operator} $\Sigma $ defined by
\begin{equation}
\Sigma = W_+^* W_- .
\label{S1_def_Sop}
\end{equation}
In view of Theorem \ref{S1_thm_suzuki}, we can see $\Sigma \psi_- = \psi _+ $.
As has been derived by Morioka \cite{Morioka}, the Fourier transform of the scattering operator $\Sigma $ can be decomposed as the direct integral
$$
 \int _{0}^{2\pi} \oplus \widehat{\Sigma} (\theta )d\theta ,
$$
where $\widehat{\Sigma} ( \theta )$ is a unitary operator on a Hilbert space $ {\bf h} (\theta )$ depending on the spectral parameter $ \theta$.
In the following, we call $\widehat{ \Sigma} (\theta )$ the scattering matrix or the \textit{S-matrix} for short.
The spectral decomposition associated with $ U_0 $ or $U$ are defined as a distorted Fourier transform.
For details, we will discuss in Section 2. 
We also mention the stationary measure of QWs.
Generalized eigenfunctions of QWs give stationary measures.
For the topic of stationary measures, see Konno \cite{Konno}, Konno-Takei \cite{KonnoTakei}, Komatsu-Konno \cite{KomatsuKonno}, Kawai et al. \cite{KawaiKomatsuKonno}, and so on.
%In this paper, we give an explicit expression of the scattering in this setting. 

The purpose of this paper is to derive an explicit expression of the scattering operator for $U$ and $U_0$.　
In fact, we derive a construction of $\Sigma$ and $\widehat{\Sigma} (\theta )$.
For the case where $1\leq M \leq 3$, see Section 3.
The primitive form of QWs found in Feynman's checker board \cite{FH} is useful to our combinatorial approach.
In Section 4, we consider the general case.
Our construction is based on the counting paths up i.e. this is a time-dependent method.
The Fourier transform of $\Sigma $ gives an explicit formula of the S-matrix determined by the matrix $C_x$.
Moreover, our arguments are deeply related with the resonant-tunneling of QWs (see \cite{MMOS}).
Precisely, a generalized eigenfunction of $U$ will be constructed in $ \ell^{\infty} ( \mathbb{Z} ; \mathbb{C}^2 )$ (see \cite{HS}), and the S-matrix $\widehat{\Sigma} (\theta )$ appears in this eigenfunction.

%%%%%%%%%%%%%%%%%%%%%%%%%%%%%%

\medskip

Notations which will be used in this paper are as follows.
$\mathbb{T}=\mathbb{R}/(2\pi \mathbb{Z})$ denotes the flat torus.
We often identify $\mathbb{T}$ with $[0,2\pi )$ or $[-\pi ,\pi )$ modulo $2\pi $.
For a sequence $f=\{ f(x)\} _{x\in \mathbb{Z}} $, we define
$$
\widehat{f} (\xi )=\frac{1}{\sqrt{2\pi}} \sum _{x\in \mathbb{Z}} e^{-ix\xi} f(x) , \quad \xi \in \mathbb{T} .
$$
On the other hand, for a distribution $\widehat{g}$ on $\mathbb{T}$, we define the Fourier coefficient $g(x)$ for $x\in \mathbb{Z}$ by
$$
g(x)=\frac{1}{\sqrt{2\pi}} \int _{\mathbb{T}} e^{ix\xi} \widehat{g} (\xi )d\xi .
$$
Note that the mapping $  f \mapsto \widehat{f}$ is a unitary operator from $ \ell^2 (\mathbb{Z} ) $ to $  L^2 (\mathbb{T} )$.
The unitary mapping $\mathcal{U} : \mathcal{H} \to \widehat{\mathcal{H}} := L^2 (\mathbb{T} ;\mathbb{C}^2 )$ 
is defined by $ \mathcal{U} ( [ f_L , f_R ]^{\mathsf{T}}) =  [ \widehat{f}_L , \widehat{f}_R ]^{\mathsf{T}}$.

%%%%%%%%%%%%%%%%%%%%%%%%%%%%%%%%%

\section{Spectral property for QW}
\subsection{Spectra}

Let us recall some basic notions of spectra for unitary operators.
Let $E_U ( \theta ) $ be the spectral decomposition of $U$.
Since $E_U (\theta )$ is a measure on $\mathbb{R}$, applying the Radon-Nikod\'{y}m theorem, it provides the orthogonal decomposition of $ \mathcal{H}$ associated with $U$ as 
$$
\mathcal{H}= \mathcal{H} _p (U) \oplus \mathcal{H}_{ac} (U) \oplus \mathcal{H}_{sc} (U) ,
$$
where $ \mathcal{H}_p (U)$ is the closure of all eigenspaces of $U$, $ \mathcal{H}_{ac} (U) $ and $ \mathcal{H}_{sc} (U)$ are the subspaces of $f\in \mathcal{H}$ such that $ (E_U (\theta )f,f)_{\mathcal{H}}$ is absolutely continuous or singular continuous with respect to $ \theta $, respectively.
Then the spectrum $ \sigma (U)$ is classified by
$$
\sigma_p (U)= \text{the set of eigenvalues of } U,
$$
$$
\sigma_{ac} (U) = \sigma (U | _{\mathcal{H}_{ac} (U)} ) ,\quad \sigma_{sc} (U)= \sigma (U| _{\mathcal{H}_{sc} (U)} ) .
$$

Another classification of $\sigma (U)$ is given in view of the sets of points on the complex plane.
The discrete spectrum $\sigma_d (U)$ is the set of isolated eigenvalues of $U$ with finite multiplicities.
The essential spectrum $\sigma_{ess} (U)$ is defined by $\sigma (U)\setminus \sigma_d (U)$.
If $ \lambda \in \sigma_{ess} (U)$, $\lambda$ is either an eigenvalue with infinite multiplicity or an accumulation point of $\sigma (U)$.

Let us turn to the QWs $U$ and $U_0 $.
We put $ \widehat{U}_0 = \mathcal{U} U_0 \mathcal{U}^* $.
Then $ \widehat{U}_0$ is the operator of multiplication by the matrix
$$
\widehat{U}_0 (\xi )= \left[ \begin{array}{cc}
e^{i\xi} & 0 \\ 0 & e^{-i\xi} \end{array} \right] , \quad \xi \in \mathbb{T} .
$$
We obtain for any $ \lambda \in \mathbb{C}$ 
$$
\det (\widehat{U} _0 (\xi )-\lambda )= \lambda^2 -2 \lambda \cos \xi +1 . 
$$

\begin{lemma}
We have $ \sigma (U_0) = \sigma_{ac} (U_0)= \{ e^{i\theta} \ ; \ 0\leq \theta <2\pi \} $.
\label{S2_lem_spec0}
\end{lemma}

Let $ p (\xi , \theta )= \det ( \widehat{U}_0 (\xi )-e^{i\theta } )$.
If $ e^{i\theta} = \pm 1$, we have
$$
\frac{\partial p}{\partial \xi} (\xi , \theta )= 2e^{i\theta} \sin \xi =0,
$$
for $\xi$ such that $p(\xi , \theta )=0$.
For the scattering theory, we need the set
\begin{equation}
M(\theta )=\{ \xi \in \mathbb{T} \ ; \ p(\xi ,\theta )=0\} =\{ \arccos (\cos \theta ) , 2\pi - \arccos (\cos \theta )\} .
\label{S2_def_Mtheta}
\end{equation}
Note that $\arccos (\cos \theta )\in [0,\pi] $ even though $\theta$ varies on $[0,2\pi  ]$ by the definition of the principal value of $\arccos$.\footnote{In this paper, we consider the free QW as $U_0$.
Then we can simply write $M(\theta )=\{ \theta , 2\pi - \theta \} $ when we identify ${\bf T}$ with $[0,2\pi )$.
For general cases which are considered in \cite{Morioka}, we defined $M(\theta )$ by using $\arccos$ as above.}
Letting $ \mathcal{T} = \{ \pm 1 \} $, $M(\theta )$ is the set of non-degenerate zeros of $p(\xi , \theta )$ for $ e^{i\theta} \in \sigma_{ac} (U_0) \setminus \mathcal{T}$.

Since $U-U_0$ is finite rank, we can see the invariance of the essential spectrum as follows (see \cite{MoSe}).
Moreover, there is no singular continuous spectrum in our setting (see \cite{Morioka}, \cite{MS4}).
\begin{lemma}

We have $ \sigma_{ess} (U)= \sigma_{ess} (U_0)= \{ e^{i\theta} \ ; \ 0\leq \theta < 2\pi \} $ and $ \sigma_{sc} (U)=\emptyset $.

\label{S2_lem_essspec}
\end{lemma}

%%%%%%%%%%%%%%%%%%%%%%%%%%%%%%

\subsection{Spectral decomposition}

For the explicit expression for the S-matrix, we use the spectral decomposition associated with $U_0$.
In our case, we can derive the spectral decomposition precisely by using the eigenvectors of $\widehat{U}_0 (\xi )$.

Suppose $e^{i\theta} \in \sigma_{ac} (U_0 )\setminus \mathcal{T} $.
For each $\xi \in \mathbb{T}$, the eigenvalues of $ \widehat{U}_0 (\xi )$ are $\lambda_{+} (\xi ) = e^{i \theta (\xi )} $ and $ \lambda_- (\xi )=e^{i(2\pi -\theta (\xi ) )} $ for $ \theta (\xi )= \arccos (\cos \xi ) \in [0,\pi ]$.
The associated orthonormal eigenvectors are given by 
$$
{\bf e}_{+} = \left[ \begin{array}{c}
1 \\ 0 \end{array} \right]  \quad \text{for} \quad \lambda_+ ( \xi ) , \quad {\bf e} _{-} = \left[ \begin{array}{c}
0 \\ 1 \end{array} \right] \quad \text{for} \quad \lambda_- (\xi ).
$$
Then we define the orthogonal projections $ P_{\pm} (\xi )$ on the the eigenspaces spanned by ${\bf e}_{\pm} $, respectively.
In particular, we have 
$$
\widehat{U}_0 (\xi )= \lambda_+ (\xi ) P_+ (\xi )+ \lambda_- (\xi )P_- (\xi ) , \quad e^{i\theta} \in \sigma_{ac} (U_0 ) \setminus \mathcal{T} .
$$

Now we introduce a change of the variable $\xi \in \mathbb{T}$.
Let
$$
J_+ = [0,\pi ] =\{ \theta (\xi ) \ ; \ \xi \in \mathbb{T}^d \} , \quad J_- =  [\pi ,2\pi ]= \{ 2\pi- \theta (\xi ) \ ; \ \xi \in \mathbb{T}^d \} .
$$
In view of (\ref{S2_def_Mtheta}), we put $\xi _+ ( \theta ) = \arccos (\cos \theta )  $ and $\xi_- (\theta )=2\pi -\arccos (\cos \theta )  $.
Thus we have 
\begin{gather}
\begin{split}
\int _{\mathbb{T}} |\widehat{f} (\xi ) |^2 d\xi &  =   \int _{\mathbb{T}} |P_+ (\xi ) \widehat{f} (\xi )|^2 d\xi +  \int _{\mathbb{T}} |P_- (\xi ) \widehat{f} (\xi )|^2 d\xi  \\
= & \, \sum_{s=+,-} \int _{0}^{\pi} |P_+ (\xi_s (\theta )) \widehat{f} (\xi_s (\theta )) |^2 d\theta \\
& \, + \sum_{s=+,-} \int _{\pi}^{2\pi} |P_- (\xi_s (\theta ) )\widehat{f} (\xi_s (\theta )) |^2 d\theta .
\end{split}
\label{S2_eq_changevariable}
\end{gather}
This observation allow us to introduce the Hilbert space $\widetilde{{\bf h} } (\theta )$ which is the space of $\mathbb{C}$-valued functions on $M(\theta )$ with the inner product
$$
(\psi , \phi )_{\widetilde{{\bf h}} (\theta )} = \sum _{s= +,-} \psi (\xi_s (\theta )) \overline{\phi (\xi_s (\theta ))} .
$$
Noting that the projection $P_{\pm} (\xi )$ is given by $ P_{\pm} (\xi ) \widehat{f} (\xi )= (\widehat{f}(\xi ), {\bf e} _{\pm} )_{\mathbb{C}^2} {\bf e} _{\pm} $, we define the Hilbert space 
$$
{\bf H}_{\pm} = L^2 (J_{\pm} ; \widetilde{{\bf h}} (\theta ) {\bf e} _{\pm} ; d \theta ) , \quad \widetilde{{\bf h}} (\theta ) {\bf e} _{\pm} =\{ \phi {\bf e}_{\pm} | _{M(\theta )} \ ; \ \phi \in \widetilde{{\bf h}} (\theta ) \}  .
$$

We also introduce the Hilbert space 
$$
{\bf h} (\theta )= \widetilde{{\bf h}} (\theta ){\bf e}_+ \oplus \widetilde{{\bf h}} (\theta ){\bf e}_- ,
$$
with the inner product
$$
(\phi ,\psi )_{{\bf h} (\theta )} =  ( \phi_+ , \psi_+  )_{\widetilde{{\bf h}} (\theta )} + ( \phi_- , \psi_-  )_{\widetilde{{\bf h}} (\theta )} ,  
$$
where $ \phi = \phi_+ {\bf e}_+ \oplus \phi_- {\bf e}_-$ and $ \psi = \psi_+ {\bf e}_+ \oplus \psi_- {\bf e}_-$.
Then the distorted Fourier transform $ \widehat{\mathcal{F}}_0 (\theta )$ is defined as follows.
Let $ \widehat{\mathcal{F}}_0 (\theta )=( \widehat{\mathcal{F}}_{0,+} (\theta ) , \widehat{\mathcal{F}}_{0,-} (\theta ))$ by
$$
\widehat{\mathcal{F}}_{0,s} (\theta )\widehat{f} = \left\{
\begin{split}
P_s \widehat{f} \big| _{M(\theta )} , & \quad \theta \in J_s ,\\
0, & \quad \theta \not\in J_s ,
\end{split}
\right.
$$
for $s=+$ or $-$.
Then we have $ \widehat{\mathcal{F}}_0 (\theta )\widehat{f} \in {\bf h} (\theta )$ for $\mathbb{C}^2$-valued smooth functions $\widehat{f}$ on $\mathbb{T}$.

Let
$$
(\widehat{\mathcal{F}}_0 \widehat{f})(\theta )= \widehat{\mathcal{F}}_0 (\theta )\widehat{f} , \quad \theta \in J_+ \cup J_- .
$$

\begin{lemma}
Let ${\bf H} = {\bf H}_+ \oplus {\bf H}_- $.
Then $\widehat{\mathcal{H}} $ and ${\bf H}$ are unitary equivalent in the sense of (\ref{S2_eq_changevariable}).
The Fourier transform $\widehat{\mathcal{F}_0}$ can be extended uniquely to a unitary operator from $\widehat{\mathcal{H}}$ to ${\bf H}$.
\label{S2_lem_isometry}
\end{lemma}

Now let us turn to the scattering operator $\Sigma $.
Letting $ \mathcal{F}_0 = \widehat{\mathcal{F}}_0 \mathcal{U}$, we define the S-matrix $\widehat{\Sigma} (\theta )$ by
\begin{equation}
\widehat{\Sigma} = \mathcal{F}_0 \Sigma \mathcal{F}_0^* = \int _{J_+ \cup J_+ } \oplus \widehat{ \Sigma} (\theta )d\theta .
\label{S2_eq_smatrix}
\end{equation}
The following property has been given in Theorem 5.3 in \cite{Morioka}.

\begin{lemma}
(1) For $f\in {\bf H}$, we have $(\widehat{\Sigma}f )(\theta )= \widehat{\Sigma} (\theta )f (\theta )$ for $ \theta \in ( J_+ \cup J_- ) \setminus \{ 0,\pi \}$. \\
(2) $\widehat{\Sigma} (\theta )$ is unitary on ${\bf h} (\theta )$ for $ \theta \in ( J_+ \cup J_- ) \setminus \{ 0,\pi \}$. \\
%(3) There exists a compact operator $A(\theta )$ on ${\bf h} (\theta )$ such that 
%$$
%\widehat{\Sigma} (\theta )=1-2\pi e^{i\theta} A(\theta ) , \quad \theta \in ( J_+ \cup J_- ) \setminus \{ 0,\pi \} .
%$$
\label{S2_lem_Smatrix}
\end{lemma}

%%%%%%%%%%%%%%%%%%%%%%%%
%%%%%%%%%%%%%%%%%%%%%%%%
\section{Explicit expressions of S-operator for $M=1,2,3$}
%Let $U$ be the time evolution operator with inpurities described by (\ref{inpurity}), 
%while $U_0$ be the time evolution operator of the free walk. 
%Then from the general theory on the scattering, there exist so called wave operators $W_\pm$ on $\ell^2(\mathbb{Z};\mathbb{C}^2)$ such that 
%	\begin{equation}\label{eq:wave} 
%        W_\pm=\lim_{t\to\pm\infty}U^{-t}U_0^{t}. 
%        \end{equation}
The scattering operator in the impurity $\Gamma_M$ is denoted by 
	\[ \Sigma_M=W_+^*W_-. \]
We compute $\Sigma_M$ explicitly and reconsider the meaning of the scattering of quantum walks.
We put $W_{-,t}:=U^{t}U_0^{-t}$ and $W_{+,s}^*:=U_0^{-s}U^s$. 
In this paper, we consider 
	\[ \Sigma_M=\lim_{t,s\to\infty} U_0^{-t}U^{t+s}U_0^{-s}\] 
since 
	\begin{align*}
        || (U_0^{-t}U^{t+s}U_0^{-s}-W_+^{*}W_-)\psi ||
        	&= || ((U_0^{-t}U^{t}-W_+^{*})U^sU_0^{-s}+  W_+^*(U^sU_0^{-s}-W_-) )\psi || \\
                &\leq || (U_0^{-t}U^{t}-W_+^{*})U^sU_0^{-s} \psi || + || W_+^*(U^sU_0^{-s}-W_-) \psi || 
        \end{align*}
which implies, we have 
	\[ \lim_{t,s\to\infty}|| (U_0^{-t}U^{t+s}U_0^{-s}-W_+^{*}W_-)\psi || =0 \]
by Theorem~\ref{S1_thm_suzuki}.
%%%%%%%
\subsection{The case for $M=1$}
Let us consider the most simple case $M=1$ in the following.  
See Fig.~\ref{Fig_M1_scattering}. 
To consider $\Sigma_1\psi$ for arbitrary $\psi$, we set $\psi$ as following four cases; 
(i) $\delta_{x}|R\rangle$ for $x\leq 0$, (ii) $\delta_{x}|R\rangle$ for $x>0$, 
(iii) $\delta_{x}|L\rangle$ for $x\leq 0$, and (iv) $\delta_{x}|L\rangle$ for $x>0$. 
Let us see $\Sigma_1\delta_x|R\rangle=d\delta_{-|x|}|R\rangle+b\delta_{|x|}|L\rangle$ for $x\leq 0$ in the following. 
Since 
\[\Sigma_1=\lim_{s\to\infty}U_0^{-s}U^{s}\lim_{t\to\infty}U^{t}U_0^{-t}, \]
we compute it step by step as follows:
$\delta_x|R\rangle\stackrel{U_0^{-t}}{\longmapsto} \Psi_1 \stackrel{U^{t}}{\longmapsto} \Psi_2 
\stackrel{U^{s}}{\longmapsto} \Psi_3 \stackrel{U_0^{-s}}{\longmapsto}\Psi_4$ in the limit of $t,s\to \infty$. 
\begin{enumerate}
\item Since $U_0$ describes the free walk, $U_0^{-t}\delta_x|R\rangle=\delta_{-|x|-t}|R\rangle$ holds. 
\item The walker in the inpurity starting from $\delta_{-|x|-t}|R\rangle$ proceeds freely toward the origin but clearly 
never exceed the origin, where the inpurity is set, within time $t$. 
Then the walker behaves a free walker, which implies $W_{-,t}\delta_x|R\rangle=U^tU_0^{-t}\delta_x|R\rangle=\delta_x|R\rangle$. 
Therefore $W_-\delta_x|R\rangle=\delta_x|R\rangle$ for $x\leq 0$. 
\item In the next, from this initial state, let us consider $U^{s}$ for large $s$. 
Until $|x|$ steps, the walk from $\delta_{-|x|}|R\rangle$ is a free walking and 
at time $|x|$ step, the walker finally reaches to the origin, and in the next time, the walker is spitted into two walkers, in other word, 
described by a linear combination of $\delta_{-1}|L\rangle$ and $\delta_{1}|R\rangle$ by the reflection and transmission; 
$U^{|x|+1}\delta_{-|x|}|R\rangle=b\delta_{-1}|L\rangle+d\delta_{1}|R\rangle$. 
After time $|x|+2$, the walk again becomes free. 
Since $s$ is sufficiently large, which implies $|x|<s$, 
we have $U^s\delta_{x}|R\rangle= b\delta_{-(s-|x|)}|L\rangle+d\delta_{s-|x|}|R\rangle$
\item Since $U_0$ describes a free walking, 
\[ U_0^{-s}(b\delta_{-(s-|x|)}|L\rangle+d\delta_{s-|x|}|R\rangle)=b\delta_{|x|}|L\rangle+d\delta_{-|x|}|R\rangle. \]
Therefore $W_{+,s}\delta_x|R\rangle= b\delta_{|x|}|L\rangle+d\delta_{-|x|}|R\rangle$ for sufficiently large $s$. 
Then we have $\Sigma\delta_x|R\rangle=b\delta_{|x|}|L\rangle+d\delta_{-|x|}|R\rangle$. 
\end{enumerate}

Let us see $\Sigma_1\delta_x|R\rangle=d\delta_{x}|R\rangle+b\delta_{-x}|L\rangle$ for $x>0$ in the following. 
\begin{enumerate}
\item Since $U_0$ describes the free walk, $U_0^{-t}\delta_x|R\rangle=\delta_{-t+x}|R\rangle$ holds. 
Remark that since $t$ is sufficiently large, $-t+x<0$. 
\item The walker in the inpurity starting from $\delta_{-t+x}|R\rangle$ proceeds to the origin freely until time $x$ 
and at time $x$ she hits the origin and splits into left and right directions; that is, 
	\[ U^{x+1}\delta_{-t+x}|R\rangle=b\delta_{-1}|L\rangle+d\delta_{1}|R\rangle. \] 
After $x+1$ step, the walk becomes free again. Then since $t$ is sufficiently large, 
$U^t\delta_{-t+x}|R\rangle=b\delta_{-x}|L\rangle+d\delta_{x}|R\rangle$. 
Therefore $W_-\delta_x|R\rangle=b\delta_{-x}|L\rangle+d\delta_{x}|R\rangle$ for $x>0$. 
\item In the next, from the initial state; $b\delta_{-x}|L\rangle+d\delta_{x}|R\rangle$, 
let us consider $U^s$ for large $s$. 
Since the walker never hit the origin due to the initial state, the walk behaves free. 
Then 
	\[ U^s(b\delta_{-x}|L\rangle+d\delta_{x}|R\rangle)= b\delta_{-x-s}|L\rangle+d\delta_{x+s}|R\rangle. \]
\item Since $U_0$ is free walking, we have 
	\[ U_0^{-s}(b\delta_{-x-s}|L\rangle+d\delta_{x+s}|R\rangle)=b\delta_{-x}|L\rangle+d\delta_{x}|R\rangle. \]
It means $W_+^*(b\delta_{-x}|L\rangle+d\delta_{x}|R\rangle)=(b\delta_{-x}|L\rangle+d\delta_{x}|R\rangle)$. 
Then we obtain $\Sigma_1\delta_{x}|R\rangle=b\delta_{-x}|L\rangle+d\delta_{x}|R\rangle$ for $x>0$. 
\end{enumerate}
After all, for any $x\in \mathbb{Z}$, we have
\[ \Sigma_1\delta_x|R\rangle=b\delta_{-x}|L\rangle+d\delta_{x}|R\rangle. \] 
In the same way, for any $x\in \mathbb{Z}$, we have
\[ \Sigma_1\delta_x|L\rangle=c\delta_{-x}|R\rangle+a\delta_{x}|L\rangle.\]
We summarize the computational result in the following Proposition. 
\begin{proposition}\label{PropM1}
Assume $|a|\neq 1$. 
The scattering operator $\Sigma_1$ is 
\begin{align*}
\Sigma_1\delta_x|R\rangle &= b\delta_{-x}|L\rangle+d\delta_{x}|R\rangle, \\
\Sigma_1\delta_x|L\rangle &= c\delta_{-x}|R\rangle+a\delta_{x}|L\rangle.
\end{align*}
%	\begin{equation}\label{eq:M1}
%         (\Sigma_1\psi)(x)=\begin{bmatrix} a & 0 \\ 0 & d \end{bmatrix}\psi(x)+\begin{bmatrix} 0 & b \\ c & 0 \end{bmatrix}\psi(-x) \end{equation}
\end{proposition}
%
%%%
\begin{figure}[!ht]
\begin{center}
	\includegraphics[width=100mm]{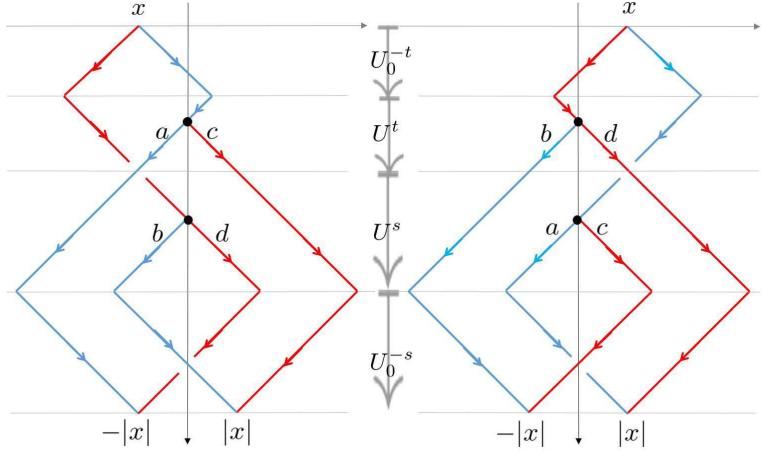}
\end{center}
\caption{ Scattering operator for $M=1$%%%
}
\label{Fig_M1_scattering}
\end{figure}
\subsection{The case for $M=2$}
Let us consider $\Sigma_2\delta_{x}|R\rangle$ with $x\leq 0$. 
Since $\Sigma_2=\lim_{t,s\to\infty}U_0^{-s}U^{t+s}U_0^{-t}$, we consider 
$\delta_x|R\rangle \stackrel{U_0^{-t}}{\longmapsto} \Psi_1 \stackrel{U^{t+s}}{\longmapsto} \Psi_2 \stackrel{U_0^{-s}}{\longmapsto}\Psi_3$ 
step by step. 
See Fig.~\ref{Fig_M2_scattering}. 
By the consideration in $M=1$ case, we need to pay {\it in advance} the effect on the invasion of $\Gamma_M$. 
We have $W_-\delta_{x}|R\rangle=\delta_{x}|R\rangle$. 
Then it is sufficient to consider $W_+^*\delta_{x}|R\rangle=\lim_{s\to\infty}U^{-s}_0U^s\delta_{x}|R\rangle$. 
First let us consider $U^s\delta_{x}|R\rangle$ for large $s$. 
At time $|x|$, the walker firstly hit $\Gamma_2=\{0,1\}$ from the negative side. 
After this time, the walk starts to be captured in $\Gamma_2=\{0,1\}$ and will escape $\Gamma_2$ to negative or positive sides soon. 
We compute the amplitude of escape $\Gamma_2=\{0,1\}$ to the negative and positive sides, respectively. 
By definition of the quantum walk represented by (\ref{def:U'}), 
the amplitude reflected by $\Gamma_2$ and never entering $\Gamma_2$ is ``$b$" and its final state at time $s$ is described by 
	\begin{equation}\label{eq:1} (U^s\delta_{x}|R\rangle)(-s+|x|)=b|L\rangle.\end{equation} 
The amplitude of the $1$st escaping to the negative side is described by ``$dba$" and its final state at time $s$ is described by 
	\begin{equation*} (U^s\delta_{x}|R\rangle)(-s+|x|+2) = dba|L\rangle.\end{equation*}
In general, for $k=1,2,\dots$, 
the amplitude of the $k$-th escaping to the negative side is described by ``$d(bc)^{k-1}ba$" and its final state at time $s$ is described by 
	\begin{equation}\label{eq:2} (U^s\delta_{x}|R\rangle)(-s+|x|+2k)=d(bc)^{k-1}ba|L\rangle.\end{equation}
On the other hand, 
for $\ell=1,2,\dots$, 
the amplitude of the $\ell$-th escaping to the positive side is described by ``$d(bc)^{\ell-1}d$" and its final state at time $s$ is described by 
	\begin{equation}\label{eq:3} (U^s\delta_{x}|R\rangle)(s-|x|-2(\ell-1))=d^2(bc)^{\ell-1}|R\rangle.\end{equation}
The amplitude never escaping $\Gamma_2$ are 
	\[ \begin{cases} d(bc)^{\frac{s-|x|}{2}-1}b\delta_0|L\rangle & \text{: $s+x$ is even,} \\  
        d(bc)^{\frac{s-|x|-1}{2}}\delta_1|R\rangle & \text{: $s+x$ is odd.} \end{cases} \]
Since $U_0^{-s}\delta_x|L\rangle=\delta_{x+s}|L\rangle$ and $U_0^{-s}\delta_x|R\rangle=\delta_{x-s}|R\rangle$, 
we obtain 
	\begin{align}
        %(\Sigma_2\delta_{x}|R\rangle)(|x|) &= b|L\rangle, \\
        %(\Sigma_2\delta_{x}|R\rangle)(|x|+2k) &= d(bc)^{k-1}a|L\rangle, \;(k=1,2,\dots)\\
        %(\Sigma_2\delta_{x}|R\rangle)(-|x|-2(\ell-1)) &= d^2(bc)^{\ell-1}|R\rangle \\ \;(\ell=1,2,\dots)
        (\Sigma_2\delta_{x}|R\rangle)(y) = \lim_{s\to\infty}\Psi_3(y)=
        	\begin{cases}
                adb(bc)^{\frac{y-|x|}{2}}|L\rangle & \text{: $y\geq |x|+2$ and $y-|x|$ is even,}\\
                b|L\rangle & \text{: $y=|x|$,} \\
                d^2(bc)^{\frac{|y|-|x|}{2}}|R\rangle & \text{: $y\leq -|x|$ and $y-|x|$ is even,}\\
                0 & \text{: otherwise,}
                \end{cases}
        \end{align}
for $x\leq 0$. 

\noindent \\
Next, let us consider $\Sigma\delta_{x}|R\rangle$ with $x>0$. 
We have $\Psi_1=U^{-t}_0 \delta_x|R\rangle=\delta_{x-t}|R\rangle$. 
Since $t$ is sufficiently large, $x-t<0$. 
Let us consider $U^{t+s}\Psi_1$. 
After $(t-x)$-step, the walker is captured in $\Gamma_2$ and will escapes $\Gamma_2$ soon to the negative or positive sides. 
The amplitude of escaping $\Gamma_2$ can be expressed by just changing $x\to -t+x$ and $s\to t+s$ in (\ref{eq:1})--(\ref{eq:3}); that is, 
	\begin{align}
        (U^{t+s}\Psi_1)(-s-x) &= b|L\rangle, \\
        (U^{t+s}\Psi_1)(-s-x+2k) &= adb(bc)^{k-1}|L\rangle, \;\;(k=1,2,\dots) \\
        (U^{t+s}\Psi_1)(s+x-2(\ell-1)) &=d^2(bc)^{\ell-1}|R\rangle. \;\;(\ell=1,2,\dots)
        \end{align}
The amplitudes never escaping $\Gamma_2$ are $d(bc)^{\frac{s+x}{2}-1}b\delta_0|L\rangle$ if $s+x$ is even, while 
$a(bc)^{\frac{s+x-1}{2}}\delta_1|R\rangle$. 
It is equivalent to 
	\begin{multline*} 
        \Psi_2=U^{t+s}\Psi_1
        \\=\left(b\delta_{-s-x}+\sum_{k\in K_1}adb(bc)^{k-1}\delta_{-s-x+2k}\right)|L\rangle 
        + \sum_{\ell\in K_2}d^2(bc)^{\ell-1}\delta_{s+x-2(\ell-1)}|R\rangle \\ 
        + 
        \begin{cases}
        d(bc)^{\frac{s+x}{2}-1}b\delta_0|L\rangle & \text{: $s+x$ is even,}\\
        d(bc)^{\frac{s+x-1}{2}}\delta_1|R\rangle & \text{: $s+x$ is odd.}
        \end{cases}
         \end{multline*}
Here $K_1,K_2\subset \mathbb{N}$ is 
	\[ K_1=\begin{cases} 1\leq k \leq \frac{s+x}{2}-1 & \text{: $s+x$ is even,}\\ 1\leq k\leq \frac{s+x-1}{2} & \text{: $s+x$ is odd,} \end{cases} \]
	\[ K_2=\begin{cases} 1\leq \ell \leq \frac{s+x}{2} & \text{: $s+x$ is even,}\\ 1\leq \ell\leq \frac{s+x-1}{2} & \text{: $s+x$ is odd,} \end{cases} \]
respectively. 
Then we have 
	\[ \lim_{t,s\to\infty}\Psi_3=\Sigma_2\delta_{x}|R\rangle 
        = \left(b\delta_{-x}+\sum_{k\geq 1}adb(bc)^{k-1}\delta_{-x+2k}\right)|L\rangle 
        + \sum_{\ell\geq 1}d^2(bc)^{\ell-1}a\delta_{x-2(\ell-1)}|R\rangle. \]
Then for any $y$ with $x+y\in$even,  
	\[ (\Sigma_2\delta_x|R\rangle)(y)
        	=\begin{cases}
                d^2(bc)^{\frac{x-y}{2}}|R\rangle & \text{: $y<-x$,} \\
                b|L\rangle + d^2(bc)^{\frac{x-y}{2}}a|R\rangle & \text{: $y=-x$,} \\
                adb(bc)^{\frac{x+y}{2}-1}|L\rangle + d^2(bc)^{\frac{x-y}{2}}a|R\rangle & \text{: $-x+2\leq y\leq x$,} \\
                adb(bc)^{\frac{x+y}{2}-1}|L\rangle & \text{: $y>x$}  
                \end{cases}
                 \]
for $x>0$. 
Let us put $(\Sigma_2\delta_x|R\rangle)(y)=[(\Sigma_2\delta_x|R\rangle)_L(y),(\Sigma_2\delta_x|R\rangle)_R(y)]^\top$. 
Then the supports $(\Sigma_2\delta_x|R\rangle)_L(y)$ and $(\Sigma_2\delta_x|R\rangle)_R(y)$ are $\{-x+2k \;|\; k=0,1,\dots\}$ and $\{x-2\ell \;|\; \ell=0,1,\dots\}$, respectively and  
	\begin{align}
        (\Sigma_2\delta_x|R\rangle)_L(-x+2k) &= \begin{cases} b & \text{: $k=0$}\\ adb(bc)^{k-1} & \text{: $k=1,2,\dots$} \end{cases} \\
        (\Sigma_2\delta_x|R\rangle)_R(x-2\ell) &= d^2(bc)^{\ell}\;\;(\ell=0,1,2,\dots)
        \end{align}
In the same way, 
	the supports $(\Sigma_2\delta_x|L\rangle)_L(y)$ and $(\Sigma_2\delta_x|L\rangle)_R(y)$ are $\{-x+2+2k \;|\; k=0,1,\dots\}$ and 
        $\{x-2\ell \;|\; \ell=0,1,\dots\}$, respectively and  
	\begin{align}
        (\Sigma_2\delta_x|L\rangle)_R(-x+2-2k) &= \begin{cases} c & \text{: $k=0$}\\ acd(cb)^{k-1} & \text{: $k=1,2,\dots$} \end{cases} \\
        (\Sigma_2\delta_x|L\rangle)_L(x-2\ell) &= a^2(cb)^{\ell}\;\;(\ell=0,1,2,\dots)
        \end{align}
We summarize the above computational results in the following Proposition. 
\begin{proposition}\label{PropM2}
For $M=2$, the scattering operator $\Sigma_2$ is described as follows: 
The supports of $(\Sigma_2\delta_x|R\rangle)_L(y)$ and $(\Sigma_2\delta_x|R\rangle)_R(y)$ are $\{-x+2k \;|\; k=0,1,\dots\}$ and $\{x-2\ell \;|\; \ell=0,1,\dots\}$, respectively and  
	\begin{align}
        (\Sigma_2\delta_x|R\rangle)_L(-x+2k) &= \begin{cases} b & \text{: $k=0$}\\ adb(bc)^{k-1} & \text{: $k=1,2,\dots$} \end{cases} \\
        (\Sigma_2\delta_x|R\rangle)_R(x-2\ell) &= d^2(bc)^{\ell}\;\;(\ell=0,1,2,\dots)
        \end{align}
On the other hand, 
	the supports $(\Sigma_2\delta_x|L\rangle)_L(y)$ and $(\Sigma_2\delta_x|L\rangle)_R(y)$ are 
        $\{-x+2+2k \;|\; k=0,1,\dots\}$ and $\{x-2\ell \;|\; \ell=0,1,\dots\}$, respectively and  
	\begin{align}
        (\Sigma_2\delta_x|L\rangle)_R(-x-2k+2) &= \begin{cases} c & \text{: $k=0$}\\ acd(cb)^{k-1} & \text{: $k=1,2,\dots$} \end{cases} \\
        (\Sigma_2\delta_x|L\rangle)_L(x+2\ell) &= a^2(cb)^{\ell}\;\;(\ell=0,1,2,\dots)
        \end{align}
\end{proposition}
%%%%
%The scattering matrix $\hat{\Sigma}_2(\xi):=\mathcal{F}\Sigma_2\mathcal{F}^{-1}$ can be described as follws. 
%
%
%%%%
%%%
\begin{figure}[!ht]
\begin{center}
	\includegraphics[width=70mm]{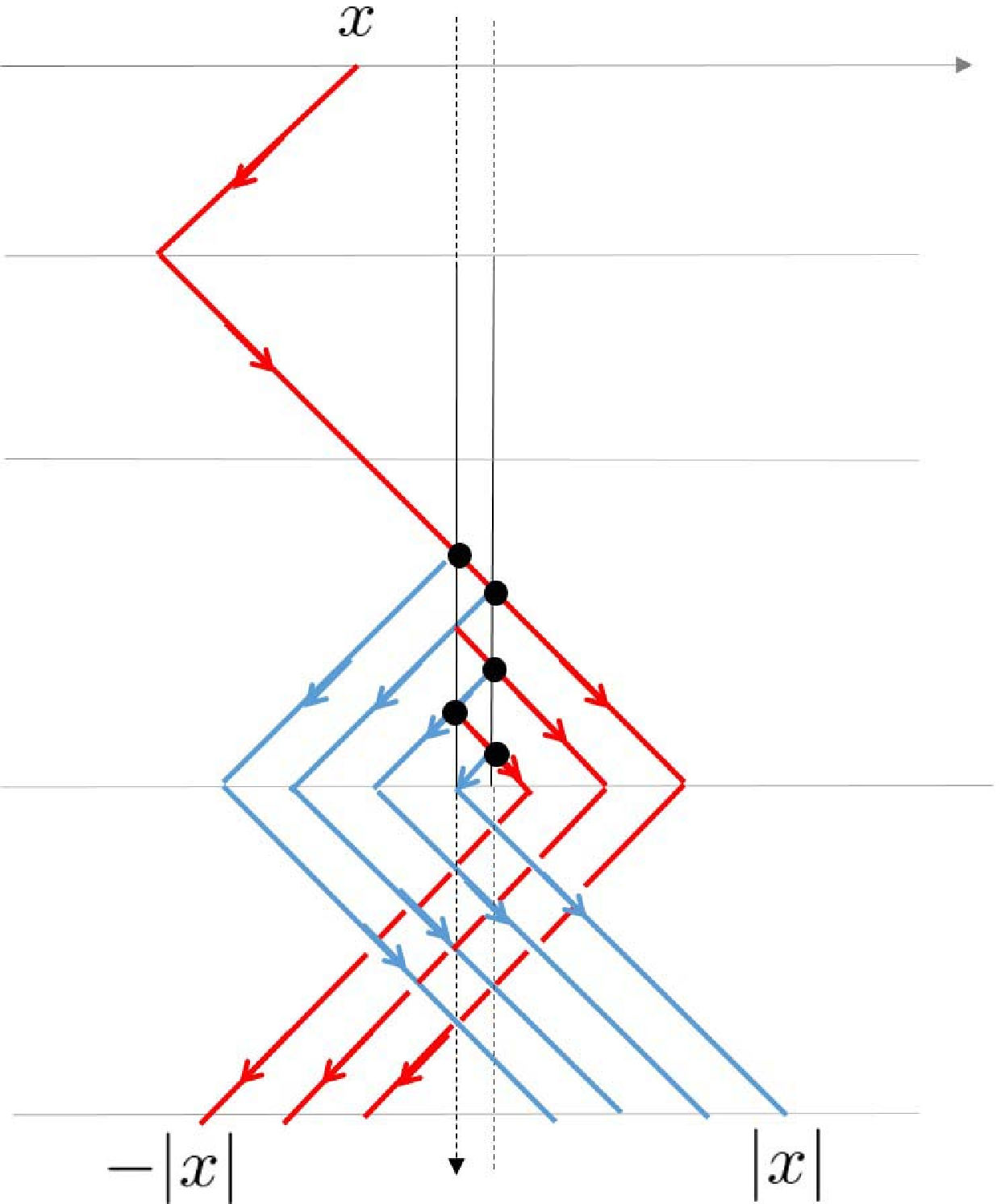}
\end{center}
\caption{ Scattering operator for $M=2$%%%
}
\label{Fig_M2_scattering}
\end{figure}
%%%%%%%%%%%%%%%%%%%%%%%%
\subsection{The case for $M=3$ case}
%%%%%%%%%%%%%%%%%%%%%%%%%
\subsubsection{Preparation}
For the preparation, let us consider the following situation. 
For $X\in \mathbb{Z}\setminus \{1\}$, the initial state is defined by 
	\[ \psi_0=\delta_X \begin{cases} |R\rangle & \text{: $X\leq 0$}\\ %\varphi & \text{: $X=1$}\\ 
        |L\rangle & \text{: $X\geq 2$} \end{cases} \] 
We consider $N$-th iteration of this quantum walk for large $N$. 
If $X\leq 0$, then the walker hits $\Gamma_3$ after $|X|$-step. 
The final position at time $N$ for the walker who never enter $\Gamma_3$ is $-N+|X|$, 
and its internal state is described by the left chirality $|L\rangle$ with some complex valued coefficient. 
The final position at time $N$ for the walker who enters $\Gamma_3$ and escapes firstly $\Gamma_3$ to the negative side is $-N+|X|+2$, 
and its internal state is described by the left chiraity $|L\rangle$ with some complex valued coefficient. 
The final position at time $N$ for the walker who enters $\Gamma_3$ and escapes secondly $\Gamma_3$ to the negative side is $-N+|X|+4$, 
and its internal state is described by the left chiraity $|L\rangle$ with some complex valued coefficient. 
Thus the support of the final position of the left chirality is $\{-N+|X|+2k \;|\; k=0,1,2,\dots\}$. 
On the other hand, 
the final position at time $N$ for the walker who enters $\Gamma_3$ and escapes firstly $\Gamma_3$ to the positive side is $N-|X|$, 
and its internal state is described by the right chiraity $|R\rangle$ with some complex valued coefficient. 
the final position at time $N$ for the walker who enters $\Gamma_3$ and escapes secondly $\Gamma_3$ to the positive side is $N-|X|-2$, 
and its internal state is described by the right chiraity $|R\rangle$ with some complex valued coefficient. 
Thus the support of the final position of the left chirality is $\{N-|X|-2k \;|\; k=0,1,2,\dots\}$. 

%Let $\phi_N(x):=\phi_N^{(L)}(x)|L\rangle +\phi_N^{(R)}(x)|R\rangle$, where $\phi_N^{(L)}(x),\phi_N^{(R)}(x)\in \mathbb{C}$. 
Putting $\Theta:=PQ+QP$, then we have 
	\begin{align}
        \phi_N(-N+|X|) &= P|R\rangle \notag \\
        \phi_N(-N+|X|+2k) &= P^2\Theta^{k-1}Q|R\rangle \;\;(k=1,2,\dots) \notag \\
        \phi_N(N-|X|-2k) &= Q^2\Theta^{k}Q|R\rangle \;\;(k=0,1,2\dots) \label{eq:negative}
        \end{align}
for $X\leq0$. 

In the same way, for $X\geq 2$ we have 
	\begin{align}
        \phi_N(N+4-X) &= Q|L\rangle \notag \\
        \phi_N(N+4-X-2k) &= Q^2\Theta^{k-1}P|L\rangle \;\;(k=1,2,\dots) \notag \\
        \phi_N(-N+X+2k) &= P^2\Theta^{k}P|L\rangle \;\;(k=0,1,2\dots) \label{eq:positive}
        \end{align}

%For $X=1$, 
%	\begin{align}
%        \phi_N^{(L)}(-N+1+2k) &= P^2\Theta^k\varphi \;\;(k=0,1,\dots) \notag \\
%        \phi_N^{(R)}(N+1-2k) &= Q^2\Theta^{k}\varphi \;\;(k=0,1,\dots) \label{eq:naka} 
%        \end{align}
        
Here $\Theta^k$ ($k=0,1,2\dots$) is described by 
	\begin{equation}\label{eq:Theta} 
        \Theta^k= (-|c|\Delta)^k \; \frac{1}{2} \left( e^{ik\theta} \begin{bmatrix} 1 & -e^{i\gamma} \\ e^{-i\gamma} & 1 \end{bmatrix} + e^{-ik\theta} \begin{bmatrix} 1 & e^{i\gamma} \\ -e^{-i\gamma} & 1 \end{bmatrix} \right)
        \end{equation}
where $\Delta=ad-bc$, $e^{i\gamma}=i\Delta |ac|/(ac)$ and $e^{i\theta}=|c|+i|a|$. 
\subsubsection{Computation of $\Sigma_3$}
Let us compute $\Psi_3:=U^sU^tU^{-t}_0\delta_x|R\rangle$ for sufficiently large $t$ and $s$. 
Remark that $U^{-t}_0\delta_x|R\rangle=\delta_{-t+x}|R\rangle$. 
Since $-t+x<0$, we apply (\ref{eq:negative}) to this by putting $X=-t+x$, $N=t+s$ and $\delta_X|R\rangle=U^{-t}_0\delta_x|R\rangle$. 
Then putting $\Psi_3(x)=[\Psi_3^{(L)}(x),\Psi_3^{(R)}(x)]^\top$, we have 
	\begin{align*}
        \Psi_3^{(L)}(-s-x) &=\langle L| P|R\rangle \notag \\
        \Psi_3^{(L)}(-s-x+2k) &= \langle L|P^2\Theta^{k-1}Q|R\rangle \;\;(k=1,2,\dots) \notag \\
        \Psi_3^{(R)}(s+x-2k) &= \langle R|Q^2\Theta^{k}Q|R\rangle \;\;(k=0,1,2\dots) 
        \end{align*}
for any $x\in \mathbb{Z}$. 
Remark that the internal states for the first, second and third of the above RHSs are $|L\rangle$, $|L\rangle$ and $|R\rangle$, respectively. 
Since $U_0^{-s}\delta_x|L\rangle=\delta_{x+s}|L\rangle$ and $U_0^{-s}\delta_x|R\rangle=\delta_{x-s}|R\rangle$, we have 
	\begin{align}
        (\Sigma_3\delta_{x}|R\rangle)_{L}(-x) &= \langle L|P|R\rangle \notag \\
        (\Sigma_3\delta_{x}|R\rangle)_{L}(-x+2k) &= \langle L|P^2\Theta^{k-1}Q|R\rangle \;\;(k=1,2,\dots) \notag \\
        (\Sigma_3\delta_{x}|R\rangle)_{R}(x-2k) &= \langle R|Q^2\Theta^{k}Q|R\rangle \;\;(k=0,1,2\dots) \label{eq:sigma3R}
        \end{align}
        
In the same way, 
let us compute $U^sU^tU^{-t}_0\delta_x|L\rangle$ for sufficiently large $t$ and $s$. 
Applying (\ref{eq:positive}) to this by putting $X=x+t$, $N=t+s$ and $\delta_X|L\rangle=U^{-t}\delta_x|L\rangle$ since $x+t>2$, 
we have 
	\begin{align}
        (\Sigma_3\delta_{x}|L\rangle)_{R}(-x+4) &= \langle R|Q|L\rangle \notag \\
        (\Sigma_3\delta_{x}|L\rangle)_{R}(-x+4-2k) &= \langle R|Q^2\Theta^{k-1}P|L\rangle \;\;(k=1,2,\dots) \notag \\
        (\Sigma_3\delta_{x}|L\rangle)_{L}(x+2k) &= \langle L|P^2\Theta^{k}P|L\rangle \;\;(k=0,1,2\dots) \label{eq:sigma3L}
        \end{align}
Combining (\ref{eq:Theta}) with (\ref{eq:sigma3R}) and (\ref{eq:sigma3L}), we obtain  the following expression of $\Sigma_3$.      
\begin{proposition}\label{PropM3}
For $M=3$, the scattering operator $\Sigma_3$ is described as follows: 
The supports of $(\Sigma_3\delta_x|R\rangle)_L(y)$ and $(\Sigma_3\delta_x|R\rangle)_R(y)$ are $\{-x+2k \;|\; k=0,1,\dots\}$ and $\{x-2k \;|\; k=0,1,\dots\}$, respectively and  
	\begin{align}
        (\Sigma_3\delta_x|R\rangle)_L(-x+2k) &= \begin{cases} b & \text{: $k=0$}\\ (-|c|\Delta)^{k-1}ad\left\{ -ie^{i\gamma}a\sin[(k-1)\theta]+b\cos[(k-1)\theta] \right\} & \text{: $k=1,2,\dots$} \end{cases} \\
        (\Sigma_3\delta_x|R\rangle)_R(x-2k) &= (-|c|\Delta)^{k-1}d^2\left\{ -ie^{i\gamma}c\sin[k\theta]+d\cos[k\theta] \right\}\;\;(k=0,1,2,\dots)
    \end{align}
On the other hand, 
	the supports $(\Sigma_3\delta_x|L\rangle)_L(y)$ and $(\Sigma_3\delta_x|L\rangle)_R(y)$ are $\{x+2k \;|\; k=0,1,\dots\}$ and $\{-x+4-2k \;|\; k=0,1,\dots\}$, respectively and  
	\begin{align}
        (\Sigma_3\delta_x|L\rangle)_R(-x+4-2k) &= \begin{cases} c & \text{: $k=0$}\\ (-|c|\Delta)^{k-1}ad\left\{ ie^{-i\gamma}d\sin[(k-1)\theta]+c\cos[(k-1)\theta] \right\} & \text{: $k=1,2,\dots$} \end{cases} \\
        (\Sigma_3\delta_x|L\rangle)_L(x+2k) &= (-|c|\Delta)^{k}a^2\left\{ ie^{-i\gamma}b\sin[k\theta]+a\cos[k\theta] \right\}\;\;(k=0,1,2,\dots)
        \end{align}
\end{proposition}
%%

%%%%%%%%%%%%%%
\section{The case for general $M$ and S-matrix}
\subsection{S-matrix}
%We put the standard basis of $\ell^2(\mathbb{Z};\mathbb{C}^2)$ by 
%$\delta_x|L\rangle=:|x;L\rangle$ and $\delta_x|R\rangle=:|x;R\rangle$. 
Recall that $|L\rangle$ and $|R\rangle$ represent the standard basis of $\mathbb{C}^2$; 
that is, $|L\rangle=[1,0]^\top$ and $R\rangle=[0,1]^\top$. 
Let $\chi: \ell^2(\mathbb{Z};\mathbb{C}^2)\to \ell^2(\Gamma_M;\mathbb{C}^2)$ be a boundary operator such that 
$(\chi\psi)(a)=\psi(a)$ for any $a\in \{(x;R),(x;L)\;|\;x\in\Gamma_M\}$. 
Here the adjoint $\chi^*: \ell^2(\Gamma_M;\mathbb{C}^2)\to \ell^2(\mathbb{Z};\mathbb{C}^2)$ is described by 
	\[ (\chi^*\varphi)(a)=\begin{cases} \varphi(a) & \text{: $a\in \{(x;R),(x;L)\;|\;x\in\Gamma_M\}$,} \\ 0 & \text{: otherwise.} \end{cases} \]
We put the principal submatrix of $U$ with respect to the impurities by $E_M:=\chi U\chi^{*}$. The matrix form of $E_M$ 
with the computational basis $\chi\delta_0|L\rangle, \chi\delta_0|R\rangle,\dots, \chi\delta_{M-1}|L\rangle, \chi\delta_{M-1}|R\rangle$ 
is expressed by the following $2M\times 2M$ matrix: 
	\begin{equation}\label{eq:E_MMat}
	     E_M= \begin{bmatrix} 
        0 & P &        &        &   \\ 
        Q & 0 & P      &        &   \\
          & Q & 0      & \ddots &   \\
          &   & \ddots & \ddots & P \\
          &   &        & Q      & 0 
        \end{bmatrix} 
	\end{equation}
We express the $((x;J),(x';J'))$ element of $E_M$ by 
    \[ (E_M)_{(x;J),(x';J')}:= \bigg\langle \chi\delta_x|J\rangle , E_M \chi \delta_{x'}|J'\rangle \bigg\rangle_{\mathbb{C}^{2M}}. \]
%Remark that $E_M$ is not a unitary matrix. 
%
First we prepare an important properties of $E_M$ as follows. 
%%%%%%%%%%%%%%%%%
\begin{lemma}\label{lemsonzai}
Let $E_{M}$ be the above {with $a\neq 0$}.$^\dagger$ Then 
$\sigma(E_{M})\subset \{ \lambda \in \mathbb{C} \;|\; |\lambda|<1 \}$. 
\end{lemma}
\begin{proof}
Let $\psi\in \ell^2(\Gamma_M,\mathbb{C}^2)$ be an eigenvector of eigenvalue $\lambda\in \sigma(E_M)$. 
Then 
	\begin{align}\label{eq:EM}  
        |\lambda|^2 ||\psi||^2
        	=|| E_M \psi ||^2 = \langle U\chi^*\psi, \chi^*\chi U\chi^*\psi \rangle 
        	\leq \langle  U\chi^*\psi, U\chi^*\psi \rangle 
                =||\chi^*\psi ||^2
                = ||\psi||^2. 
        \end{align}
Here for the inequality, we used the fact that $\chi^*\chi$ is the projection operator onto 
\[ \mathrm{span}\{ \delta_x |L\rangle, \delta_x|R\rangle \;|\; x\in \Gamma_M \}\subset \ell^2(\mathbb{Z};\mathbb{C}^2)\] 
while 
for the final equality, we used the fact that $\chi\chi^*$ is the identity operator on $\ell^2(\Gamma_M;\mathbb{C}^2)$. 
If the equality in (\ref{eq:EM}) holds, then $\chi^*\chi U\chi^*\psi=U\chi^*\psi$ holds. 
Then we have the eigenequation $U\chi^*\psi=\lambda \chi^*\psi$ by taking $\chi^*$ to both sides of 
the original eigenequation $\chi U\chi^*\psi=\lambda \psi$. 
However there are no eigenvectors having finite supports in a position independent quantum walk on $\mathbb{Z}$ with $a\neq 0$ since 
its spectrum is described by only a continuous spectrum in general. 
Thus $|\lambda|^2 < 1$.  
\end{proof}
\begin{lemma}\label{lem:diagonalizable}
$E_M$ is diagonalizable.  More precisely, if $a\neq 0$, \footnote{If $a$=0, the quantum walk becomes a trivial ``zigzag" walk because the dynamics at each vertex is only a reflection. The eigenvalues are $\pm e^{i\frac{\arg(b)+\arg(c)}{2}}$ with the multiplicities $M-1$, respectively in that case. } then
\begin{enumerate}
\item Eigenvalues of $E_M$ are simple, expect $0$. 
\item The multiplicity of the eigenvalue $0$ is $2$. 
The supports of two eigenvectors of $0$ are $\{(0;L), (0;R)\}$ and $\{(M-1;L),(M-1;R)\}$, respectively. 
\end{enumerate}
\end{lemma}
\begin{proof}
The matrix representation of $E_M$ with the permutation of the labeling such that $(x;R)\leftrightarrow (x;L)$ for any $x\in \Gamma_M$ to (\ref{eq:E_MMat}) is  
\[ E_M\cong \left[
    \begin{array}{cc|cc|cc|cc|cc}
    0 & 0 & 0 & 0 &        &   &       & & &\\ 
    0 & 0 & b & a &        &   &       & & &\\ \hline
    d & c & 0 & 0 & 0      & 0 &       & & &\\ 
    0 & 0 & 0 & 0 & b      & a &       & & &\\ \hline
     &  & d & c &  \ddots     &  &\ddots & & &\\ 
     &  & 0 & 0 &       &  &       & & &\\ \hline
      &   &   &   & \ddots &   &\ddots & & 0 & 0 \\ 
      &   &   &   &        &   &       & & b & a \\ \hline
      &   &   &   &        &   &    d &c &  0& 0\\
      &   &   &   &        &   & 0    & 0&  0& 0 
    \end{array}\right]. \]
Then the eignequation $(E_M-\lambda)\psi=0$ is expressed by 
\[ \left[
    \begin{array}{cc|cc|cc|cc|cc}
    -\lambda & 0 & 0 & 0 &        &   &       & & &\\ \hline
    0 & -\lambda & b & a &        &   &       & & &\\ 
    d & c & -\lambda & 0 &       &  &       & & &\\ \hline
     &  & 0 & -\lambda & b      & a &       & & &\\ 
     &  & d & c & -\lambda & 0 & & & &\\ \hline
     &  &  &  & \ddots &  & \ddots & & &\\ 
      &   &   &   &  &   & & &  &  \\ \hline
      &   &   &   &        &   &    0 & 0 & b & a \\ 
      &   &   &   &        &   &    d & c &  -\lambda& 0\\ \hline
      &   &   &   &        &   &     &  &  0& -\lambda 
    \end{array}\right]
   \left[
   \begin{array}{c}
        \psi(0;R) \\ \psi(0;L)  \\ \hline
        \psi(1;R) \\ \psi(1;L)  \\ \hline
        \vdots \\ \vdots \\ \hline
        \psi(M-2;R) \\ \psi(M-2;L) \\ \hline
        \psi(M-1;R) \\ \psi(M-1;L)
   \end{array}
   \right] 
    =0.  \]
Here we changed the way of blockwise of $E_M$.  
Putting 
    \[ A_\lambda :=\begin{bmatrix} 0 & -\lambda \\ d & c \end{bmatrix}, \;
       B_\lambda :=\begin{bmatrix} b & a \\ -\lambda & 0 \end{bmatrix}, \]
we have 
    \begin{align}\label{eq:rec}
        \begin{bmatrix}-\lambda & 0 \end{bmatrix}\vec{\psi}(0) &= 0, \;\; 
        A_\lambda \vec{\psi}(0) + B_\lambda \vec{\psi}(1) = 0,\;\;
        A_\lambda \vec{\psi}(1) + B_\lambda \vec{\psi}(2) = 0,\dots \notag\\
        & \dots, A_\lambda \vec{\psi}(M-2) + B_\lambda \vec{\psi}(M-1) = 0,\;\;
        \begin{bmatrix} 0 & -\lambda \end{bmatrix}\vec{\psi}(M-1) = 0,
    \end{align}
where $\vec{\psi}(x)=[\psi(x;R),\psi(x;L)]^\top$ for any $x\in \Gamma_M$. 
If $\lambda\neq 0$ and $a\neq 0$, then the inverse matrix $B_\lambda$ exists. 
We obtain
    \begin{align*}
        \psi(0;R)=0,\;
        \vec{\psi}(1)=(-B_\lambda^{-1}A_\lambda)\vec{\psi}(0),\;
        \vec{\psi}(2)=(-B_\lambda^{-1}A_\lambda)^2\vec{\psi}(0),\dots \\
        \dots \vec{\psi}(M-1)=(-B_\lambda^{-1}A_\lambda)^{M-1}\vec{\psi}(0),\;
        \psi(M-1;L)=0.
    \end{align*}
Since $\psi(0;R)=0$, $\vec{\psi}(j)$ is determined by only the one parameter $\psi(0;L)$ which implies that the dimension of the eigenspace $\lambda$ is one.  Then the eigenvalue $\lambda$ is simple if $\lambda\neq 0$. 
Note that the eigenvalue $\lambda$ for $\lambda \neq 0$ is the solution of 
    \[ \left\langle \begin{bmatrix} 0 \\ 1 \end{bmatrix},\;(-B_\lambda^{-1}A_\lambda )^{M-1}\begin{bmatrix} 0 \\ 1 \end{bmatrix} \right\rangle =0.  \]

Finally, let us consider the case for $\lambda =0$. 
(\ref{eq:rec}) for $\lambda=0$ is reduced to, 
\begin{align*}
    \vec{\psi}(0) \perp \begin{bmatrix} \bar{d} \\ \bar{c} \end{bmatrix}, & \;
    \vec{\psi}(1) \perp \begin{bmatrix} \bar{b} \\ \bar{a} \end{bmatrix},\;
    \vec{\psi}(1) \perp \begin{bmatrix} \bar{d} \\ \bar{c} \end{bmatrix},\;
    \vec{\psi}(2) \perp \begin{bmatrix} \bar{b} \\ \bar{a} \end{bmatrix},\cdots \\
    & \qquad \cdots, \vec{\psi}(M-3) \perp \begin{bmatrix} \bar{d} \\ \bar{c} \end{bmatrix},\;
    \vec{\psi}(M-2) \perp \begin{bmatrix} \bar{b} \\ \bar{a} \end{bmatrix},\;
    \vec{\psi}(M-2) \perp \begin{bmatrix} \bar{d} \\ \bar{c} \end{bmatrix},\;
    \vec{\psi}(M-1) \perp \begin{bmatrix} \bar{b} \\ \bar{a} \end{bmatrix},
    \\
\end{align*}
which implies $\vec{\psi}(x) \perp [\bar{b}\;\bar{a}]^\top,\;[\bar{d}\;\bar{c}]^\top$ for any $x=1,\dots,M-2$. 
The vectors $[\bar{b}\;\bar{a}]^\top$ and $[\bar{d}\;\bar{c}]^\top$ are linearly independent and span $\mathbb{C}^2$ space because
$[\bar{a}\;\bar{b}]^\top$ and $[\bar{c}\;\bar{d}]^\top$ are the row vectors of a unitary matrix. 
Then $\vec{\psi}(x)=0$ for any $x=1,\dots,M-2$ and 
$\vec{\psi}(0)\perp [\bar{d}\;\bar{c}]^\top$ and $\vec{\psi}(M-1)\perp [\bar{b}\;\bar{a}]^\top$, which means $\ker(E_M^k)=\ker(E_M)=
\mathbb{C}\kappa_+ \oplus 
\mathbb{C}\kappa_-$ for every $k=1,2,\dots$.  
Here 
\begin{align*} 
\kappa_+(e) =
\begin{cases} b & \text{: $e=(0;R)$} \\ 
 a & \text{: $e=(0;L)$} \\ 
0 & \text{: otherwise}
\end{cases}, 
\;\;
\kappa_-(e) =
\begin{cases} d & \text{: $e=(M-1;R)$} \\ 
 c & \text{: $e=(M-1;L)$} \\ 
0 & \text{: otherwise}
\end{cases}
\end{align*}
\end{proof}
%%%%%%%%%%%%%%%%%%%%%%

\noindent Using this finite matrix $E_M$, we give an expression of the scattering matrix as follows which is the key of our paper. 
\begin{lemma}\label{lem:scatt_M}
Let $E_M$ be the above $2M\times 2M$ matrix labeled by $\{(x;L),(x;R) \;|\; x\in \Gamma_M\}$ and $\Sigma_M$ be the scattering operator. 
Then we have 
\begin{align}
	\Sigma_M \delta_x|L\rangle 
        	&=  a\sum_{j=0}^\infty (E^{M-1+2j}_M)_{(0,L),(M-1,L)} \; \delta_{x+2j}|L\rangle   \notag \\      
                &\quad + \left(c\delta_{-x+2(M-1)} +d \sum_{\ell=1}^\infty (E^{2\ell}_M)_{(M-1,R),(M-1,L)} \; \delta_{-x+2(M-1)-2\ell}\right)|R\rangle \label{eq:scatt1} \\
	\Sigma_M \delta_x|R\rangle 
        	&=  \left(b\delta_{-x}+ a \sum_{j=1}^\infty (E^{2j}_M)_{(0,L),(0,R)} \; \delta_{-x+2j}\right)|L\rangle \notag \\    
                &\quad + d\sum_{\ell=0}^\infty (E^{M-1+2\ell}_M)_{(M-1,R),(0,R)} \;  \delta_{x-2\ell}|R\rangle  \label{eq:scatt2}  
\end{align}
\end{lemma}
\begin{proof}
For $\psi\in \ell^2(\mathbb{Z};\mathbb{C}^2)$, we describe $\psi(x)=[\psi_L(x)\;\psi_R(x)]^\top\in \mathbb{C}^2$. 
By using the arguments similar to the case where $M=1,2,3$, 
we can see that the supports of $(\Sigma_M\delta_x|L\rangle)_L(\cdot)$ and $(\Sigma_M\delta_x|L\rangle)_R(\cdot)$ are
$\{x+2j \;|\; j=0,1,2,\dots\}$ and $\{-x+2(M-1)-2\ell \;|\; \ell=0,1,2,\dots\}$, respectively. 
See Fig.~\ref{supports}.
Here the position $x+2j$, which is in the support of  $(\Sigma_M |x;L\rangle)_L(\cdot)$,  corresponds to
the position for the walker who breaks {\it into} $\Gamma_M$ from $M-1$ side, and after the $M-1+2j$ steps, breaks {\it out} $\Gamma_M$ to the opposite side $0$;  
then the amplitude is given by $a(E_M^{M-1+2j})_{(0;L),(M-1;L)}$, that is, 
\[ (\Sigma_M\delta_x|L\rangle)_L(x+2j)=a(E_M^{M-1+2j})_{(0;L),(M-1;L)}. \]
On the other hand,  
the position $-x+2(M-1)-2\ell$, which is in the support of  $(\Sigma_M |x;L\rangle)_R(\cdot)$, corresponds to the position for the walker who breaks {\it into} $\Gamma_M$ from $M-1$ side, and after the $2\ell$ steps, breaks {\it out} $\Gamma_M$ to also the same side $M-1$; 
then the amplitude is given by 
$c$ if $\ell=0$ and $d(E_M^{2\ell})_{(M-1;R),(M-1;L)}$ if $\ell>0$.
In the same way, 
the supports of $(\Sigma_M \delta_x|R\rangle)_L(\cdot)$ and $(\Sigma_M \delta_x|R\rangle)_R(\cdot)$ are
$\{-x+2j \;|\; j=0,1,2,\dots\}$ and $\{x-2\ell \;|\; \ell=0,1,2,\dots\}$. 
Here the position $-x+2j$, which is in the support of  $(\Sigma \delta_x|R\rangle)_L(\cdot)$, corresponds to
the position for the walker who breaks {\it into} $\Gamma_M$ from $0$ side, and after the $2j$ steps, breaks {\it out} $\Gamma_M$ to also the same side $0$;
then the amplitude is $b$ if $j=0$ and $a(E_M^{2j})_{(0;L),(0;R)}$ if $j>0$.
On the other hand, 
the position $x-2\ell$, which is in the support of  $(\Sigma_M \delta_x|R\rangle)_R(\cdot)$,  corresponds to
the position for the walker who breaks {\it into} $\Gamma_M$ from $0$ side, and after the $M-1+2\ell$ steps, breaks {\it out} $\Gamma_M$ to the opposite side $M-1$; 
then the amplitude is $d (E_M^{M-1+2\ell})$. 
From the above observation, we obtain the desired conclusion. 
\end{proof}
%
%%%
\begin{figure}[!ht]
\begin{center}
	\includegraphics[width=120mm]{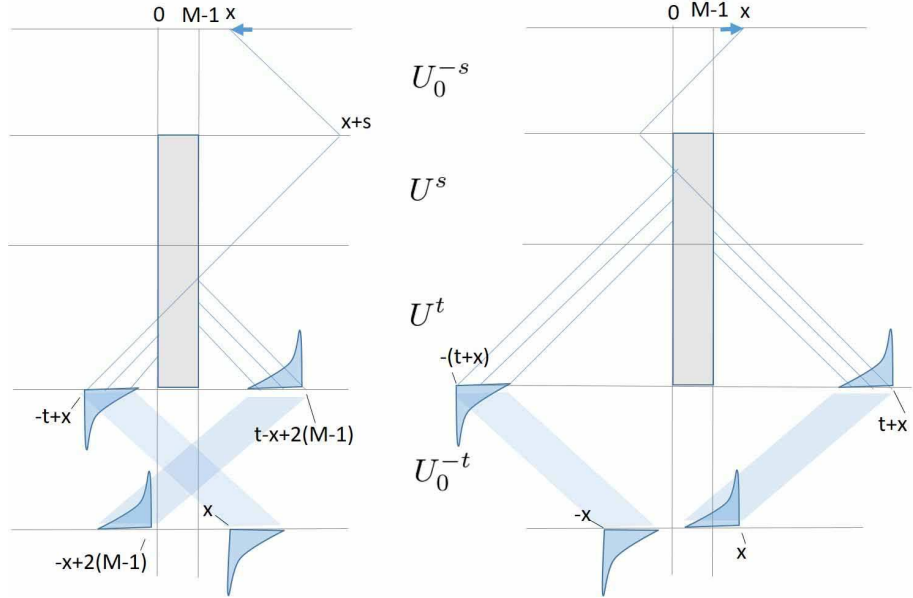}
\end{center}
\caption{ The support of the Scattering operator for general $M$%%%
}
\label{supports}
\end{figure}

\begin{remark}
Propositions~\ref{PropM1}, \ref{PropM2} and \ref{PropM3} correspond to the cases for $M=1,2,3$, in Lemma~\ref{lem:scatt_M}, respectively\footnote{If $M=1$, then $E_M={\bf 0}$. We define ${\bf 0}^0=I$. Then (\ref{eq:scatt1}) and (\ref{eq:scatt2}) recover Proposition~\ref{PropM1} for $M=1$. }. 
\end{remark}

From this fact, we are interested in the maximum absolute value of the eigenvalue of $E_M$, $r_{max,M}<1$, 
since it characterizes the scattering of this quantum walk; 
the tail's length of the exponential decay from the starting position as follows. 
\begin{theorem}
Put $\supp f=\{ x\in \mathbb{Z} \;|\; f(x)\neq 0 \}$. 
Then we have 
	\begin{align*}
        \supp  (\Sigma_M \delta_x |L\rangle)_L (\cdot) &= \{ x+2j \;|\; j=0,1,2,\dots  \} \\
        \supp  (\Sigma_M \delta_x |L\rangle)_R (\cdot) &= \{ -x+2(M-1)-2j \;|\; j=0,1,2,\dots \} \\
        \supp (\Sigma_M \delta_x |R\rangle)_L (\cdot) &= \{ -x+2j \;|\; j=0,1,2,\dots \} \\
        \supp (\Sigma_M \delta_x |R\rangle)_R (\cdot) &= \{ x-2j \;|\; j=0,1,2,\dots \}. 
        \end{align*}
Moreover, let $r_{max,M}$ be the maximum absolute value of eigenvalue of $E_M$; 
that is, $r_{max,M}:=\max\{ |\lambda| \;|\; \lambda\in \sigma(E_M) \}$. 
Then we have 
	\begin{multline*}
        r_{max,M} \\ 
        \geq \exp\left[ \lim_{j\to\infty} \log\left|(\Sigma_M  \delta_x |L\rangle)_L(x+2j)\right|^{1/j}  \right],\;\; 
        \exp\left[ \lim_{j\to\infty} \log\left|(\Sigma_M \delta_x |L\rangle)_R(-x-2j)\right|^{1/j}  \right],  \\
        \exp\left[ \lim_{j\to\infty} \log\left|(\Sigma_M \delta_x |R\rangle)_L(-x+2j)\right|^{1/j}  \right],\;\;  
        \exp\left[ \lim_{j\to\infty} \log\left|(\Sigma_M \delta_x |R\rangle)_R(x-2j)\right|^{1/j}  \right].   
        \end{multline*}
\end{theorem}
\begin{proof}
We put the eigenvalues of $E_M$  except $0$ by $\lambda_1,\lambda_2,\dots,\lambda_{2(M-1)}$, 
where $|\lambda_1|\geq |\lambda_2|\geq \cdots \geq |\lambda_{2(M-1)}|>0$.
By Lemma~\ref{lem:diagonalizable}, $(E_M^{k})_{p,q}$ can be uniquely described by 
    \[ (E_M^{k})_{p,q}= c_1\lambda_1^k+\cdots+c_{2(M-1)}\lambda_{2(M-1)}^k \]
with some complex values $c_1,\dots,c_{M-1}$. 
Assume $r_{max,M}=|\lambda_1|=\cdots=|\lambda_s|$ and at least one of $c_1, c_2, \ldots , c_s$ is not zero.  
Then 
    \[|(E_M^{k})_{p,q}|^2=r_{max}^{2k} |B_k|^2, \]
where    
    \begin{multline*}
    B_k=
    | c_1(\lambda_1/r_{max,M})^k+\cdots+c_s(\lambda_s/r_{max,M})^k \\
    +c_{s+1}(\lambda_{s+1}/r_{max,M})^k+\cdots+c_{2(M-1)}(\lambda_{2(M-1)}/r_{max,M})^k |^2
    \end{multline*}
Since there is a non-zero coefficient in $c_1,\dots,c_s$, we have 
$(1/k)\cdot \log |B_k|\to 0$ for $k\to\infty$. Then we have \[ \lim_{k\to\infty} \frac{1}{k}\log |(E_M^k)_{p,q}|^2=\log r_{max,M}^2. \]
If $c_1c_2\cdots c_s=0$, then taking $r'=|c_\kappa|(<r_{max})$, $(\kappa>s)$  
so that $c_1=0,\cdots, c_{\kappa-1}=0, c_{\kappa} \neq 0$, using the same argument as the above, we have 
\[ \lim_{k\to\infty} \frac{1}{k}\log |(E_M^k)_{p,q}|^2=\log {r'}^2. \]
Then applying it to Lemma~\ref{lem:scatt_M}, we obtain the desired conclusion.  
\end{proof}
\begin{remark}
In the cases of $M=1,2,3$, we confirm that the decay rates coincides with 
	\begin{align*}
        r_{max,1}^2 = 0, \;
        r_{max,2}^2 = |bc| = |c|^2, \;
        r_{max,3}^2 = |c|,
        \end{align*}
respectively. 
See Propositions~\ref{PropM1}, \ref{PropM2}, \ref{PropM3}, respectively. 
\end{remark}
%

%
%Let the Fourier transform of $\mathcal{U}: \ell^2(\mathbb{Z};\mathbb{C}^2)\to L^2([0,2\pi);\mathbb{C}^2)$ be 
%\[ (\mathcal{U}\psi)(\xi)=\sum_{x\in \mathbb{Z}}\psi(x)e^{i\xi x}\] 
%and the inverse $\mathcal{U}^{-1}: L^2([0,2\pi);\mathbb{C}^2)\to \ell^2(\mathbb{Z};\mathbb{C}^2)$ 
%be $(\mathcal{U}\psi)(\xi)=\sum_{x\in \mathbb{Z}}\psi(x)e^{i\xi x}$ 
%\[ (\mathcal{U}^{-1}\hat{\psi})(x)=\int_{0}^{2\pi} \hat{\psi}(\xi)e^{-i\xi x}d\xi/(2\pi). \]
Now let us proceed to the Fourier transform of the scattering operator $\mathcal{U}\Sigma_M \mathcal{U}^{-1}$. 
Let $G_M(\xi)$ be $(1-(e^{-i\xi}E_M)^2)^{-1}$. 
Note that the existence of $G_M(\xi)$ is ensured by Lemma~\ref{lemsonzai}
\begin{theorem}\label{lem:scatt_M2}
The scattering matrix $\widehat{\Sigma}_M(\theta)$ is described as follows: 
letting arbitrarily $\psi=\psi_+{\bf e}_+ +\psi_-{\bf e}_- \in {\bf h}(\theta)$ be represented by $\psi=[\psi_1 \; \psi_2]^\top$, then we have 
	\[ (\widehat{\Sigma}_M(\theta)\psi)(\xi_s(\theta))
        	= \begin{bmatrix}\widehat{a}_M(\xi_s(\theta)) & 0 \\ 0 & \widehat{d}_M(\xi_s(\theta))\end{bmatrix}\psi(\xi_s(\theta))
        	+\begin{bmatrix}0 & \widehat{b}_M(\xi_s(\theta)) \\ \widehat{c}_M(\xi_s(\theta)) & 0\end{bmatrix}\psi(-\xi_s(\theta)) ,\]
for $s\in \{ \pm \} $.
Here the values for $\widehat{a}_M(\xi)$, $\widehat{b}_M(\xi)$, $\widehat{c}_M(\xi)$, $\widehat{d}_M(\xi)$ are 
\begin{align}
\widehat{a}_M(\xi) &= a \left(E_M^{M-1}G_M(\xi)\right)_{(0;L),\;(M-1;L)},  \notag \\
\widehat{b}_M(\xi) &= b + a  \left((e^{-i\xi}E_M)^{2}G_M(\xi)\right)_{(0,L),\;(0,R)}, \notag \\
\widehat{c}_M(\xi) &= \left(c + d \left((e^{i\xi}E_M)^{2}G_M(-\xi)\right)_{(M-1;R),\;(M-1;L)}\right) e^{-2i\xi(M-1)},  \notag \\
\widehat{d}_M(\xi) &= d \left(E_M^{M-1}G_M(-\xi)\right)_{(M-1;R),\;(0;R)}. \label{eq:abcd}
\end{align}
\end{theorem}
\begin{remark}\label{remabcd}
In particular, the values $\widehat{a}_M(\xi),\widehat{b}_M(\xi),\widehat{c}_M(\xi),\widehat{d}_M(\xi)$ for $M=1,2,3$ are expressed as follows: 
	\[ \widehat{a}_1(\xi)=a,\;\;\widehat{d}_1(\xi)=d,\;\;\widehat{b}_1(\xi)=b,\;\;\widehat{c}_1(\xi)=c; \]
	\begin{align*}
        \widehat{a}_2(\xi) = \frac{a^2}{1-bce^{-2i\xi}},\;\; \widehat{d}_2(\xi) = \frac{d^2}{1-bce^{2i\xi}},\;\;
        \widehat{b}_2(\xi) = \frac{1+\Delta e^{-2i\xi}}{1-bce^{-2i\xi}}b,\;\; \widehat{c}_2(\xi) = \frac{1+\Delta e^{2i\xi}}{1-bce^{2i\xi}}ce^{-2i\xi};
        \end{align*}
        \begin{align*}
        \widehat{a}_3(\xi) &= \frac{a^3}{1-2bce^{-2i\xi}-\Delta bc e^{-4i\xi}},\;\;\;\;\;\;\;\;\;\; \widehat{d}_3(\xi) = \frac{d^3}{1-2bce^{2i\xi}-\Delta bc e^{4i\xi}},\\
        \widehat{b}_3(\xi) &= \frac{1+(\Delta-bc)e^{-2i\xi}+\Delta^2 e^{-4i\xi}}{1-2bce^{-2i\xi}-\Delta bc e^{-4i\xi}}b,\;\; 
        \widehat{c}_3(\xi) = \frac{1+(\Delta-bc)e^{2i\xi}+\Delta^2 e^{4i\xi}}{1-2bce^{2i\xi}-\Delta bc e^{4i\xi}}ce^{-4i\xi},
        \end{align*}
where $\Delta=ad-bc$. 
\end{remark}
\begin{proof}
Taking the Fourier transform to both sides (\ref{eq:scatt1}) and (\ref{eq:scatt2}) in Lemma~\ref{lem:scatt_M}, we obtain 
	\begin{align*}
        \sqrt{2\pi}(\mathcal{U}\Sigma_M\delta_x|R\rangle)(-\xi)
        	&= \left\{ be^{-i\xi x}+a(E_M^2)_{(0;L),\;(0;R)}e^{i\xi(-x+2)}+a(E_M^4)_{(0;L),\;(0;R)}e^{i\xi(-x+4)}+\cdots \right\}|L\rangle \\
                & \quad +\left\{d (E_M^{M-1})_{(M-1;R),\;(0;R)}e^{i\xi x}+ d (E_M^{M-1})_{(M-1;R),\;(0;R)}e^{i\xi (x-2)}+\cdots \right\}|R\rangle \\
                &= e^{-i\xi x} \left\{ b+ a( (e^{i\xi}E_M))^2 G_M(-\xi) )_{(0;L),\;(0;R)}  \right\} |L\rangle \\
                &\quad  + de^{i\xi x} (E_M^{M-1}G_M(\xi))_{(M-1;R),\;(0;R)} |R\rangle.  
        \end{align*}
In the same way, 
	 \begin{align*}
        \sqrt{2\pi}(\mathcal{U}\Sigma_M\delta_x|L\rangle)(-\xi)
                &= e^{i\xi (-x+2(M-1))} \left\{ c+ d ( (e^{-i\xi}E_M))^2 G_M(\xi) )_{(M-1;R),\;(M-1;L)}  \right\} |R\rangle  \\
                &\quad   + a e^{i\xi x} (E_M^{M-1}G_M(\xi))_{(0;L),\;(M-1;L)} |L\rangle.
     \end{align*}
We put for any $\hat{\varphi}\in L^2(\mathbb{T};\mathbb{C}^2)$, $\hat{\varphi}(\xi)=\widehat{\varphi}_L(\xi)|L\rangle + \widehat{\varphi}_R(\xi)|R\rangle$. 
Since $\delta_x|R\rangle=\mathcal{U}^{-1}(e^{-i\xi x}/\sqrt{2\pi})|R\rangle$, $\delta_x|L\rangle=\mathcal{U}^{-1}(e^{-i\xi x}/\sqrt{2\pi})|L\rangle$, and 
$\widehat{\varphi}_L(\xi)$, $\widehat{\varphi}_R(\xi)$ 
are described by a linear combination of $\{e^{-i\xi x}/\sqrt{2\pi} \;|\; x\in\mathbb{Z}\}$, 
we obtain
	\begin{align*} 
        \mathcal{U}\Sigma_M\mathcal{U}^{-1}\hat{\varphi}_R(-\xi)|R\rangle 
                &=  \widehat{\varphi}_R(\xi) \left\{ b+ a( (e^{i\xi}E_M))^2 G_M(-\xi) )_{(0;L),\;(0;R)}  \right\} |L\rangle \\
                &\quad  + d \widehat{\varphi}_R(-\xi)(E_M^{M-1}G_M(\xi))_{(M-1;R),\;(0;R)} |R\rangle,  \\
        \mathcal{U}\Sigma_M\mathcal{U}^{-1}\widehat{\varphi}_L(-\xi)|L\rangle 
        	&= e^{2i(M-1)\xi} \widehat{\varphi}_L(\xi) \left\{ c+ d( (e^{-i\xi}E_M))^2 G_M(\xi) )_{(M-1;R),\;(M-1;L)}  \right\} |R\rangle  \\
                &\quad   + a\widehat{\varphi}_L(-\xi) (E_M^{M-1}G_M(\xi))_{(0;L),\;(M-1;L)} |L\rangle. 
        \end{align*}
Since $|L\rangle=[1,0]^\top$ and $|R\rangle=[0,1]^\top$,  
the Fourier transform of the scattering operators $\mathcal{U}\Sigma_M \mathcal{U}^{-1}$ is expressed as follows: 
for any $\widehat{\varphi}\in L^2(\mathbb{T};\mathbb{C}^2)$, 
	\begin{equation}\label{eq:FTS} \widehat{S}(\xi)\widehat{\varphi}(\xi):=(\mathcal{U}\Sigma_M \mathcal{U}^{-1})(\xi)\widehat{\varphi}(\xi) 
        = \begin{bmatrix}\widehat{a}_M(\xi) & 0 \\ 0 & \widehat{d}_M(\xi)\end{bmatrix}\widehat{\varphi}(\xi)
        	+\begin{bmatrix}0 & \widehat{b}_M(\xi) \\ \widehat{c}_M(\xi) & 0\end{bmatrix}\widehat{\varphi}(-\xi). 
    \end{equation}
By the definition of ${\bf h}(\theta)$, we obtain the desired conclusion. 
\end{proof}
\begin{corollary}\label{cor:abcd}
The scattering operator $\Sigma_M$ can be represented by 
	\[ \left(\Sigma_M \begin{bmatrix}\psi_L \\ \psi_R \end{bmatrix}\right)(x)
	=\left(\begin{bmatrix} a_M & 0 \\ 0 & d_M \end{bmatrix}*\begin{bmatrix}\psi_L \\ \psi_R \end{bmatrix}\right)(x)
        	 +\left(\begin{bmatrix} 0 & b_M \\ c_M & 0 \end{bmatrix}*\begin{bmatrix} J\psi_L \\ J\psi_R \end{bmatrix}\right)(x) \]
for any $\psi\in \ell^2(\mathbb{Z};\mathbb{C}^2)$ with $\psi(x)=[\psi_L(x)\; \psi_R(x)]^\top$, 
where $f*g$ is the convolution of $f$ and $g$ in $\ell^2(\mathbb{Z})$ and $(J\psi)(x)=\psi(-x)$. 
Here $q_M=\mathcal{U}^{-1}\hat{q}$ for $q\in \{a,b,c,d\}$ are described by  
\begin{align*}
a_M(x) &= 
\begin{cases} 
a\; (E_M^{M-1+x})_{(0;L),(M-1;L)} & \text{: $x\geq 0$, $x$ is even}\\ 0 & \text{: otherwise,} 
\end{cases}\\
b_M(x) &= 
\begin{cases} 
b & \text{: $x=0$}\\ 
a\; (E_M^{x})_{(0;L),(0;R)} & \text{: $x>0$, $x$ is even} \\ 
0 & \text{; otherwise,}
\end{cases}\\
c_M(x) &= 
\begin{cases} 
c & \text{: $x=2(M-1)$}\\ 
d\; (E_M^{2(M-1)-x})_{(M-1;R),(M-1;L)} & \text{: $x<2(M-1)$, $x$ is even} \\ 
0 & \text{; otherwise,}
\end{cases}\\
d_M(x) &= 
\begin{cases} 
d\; (E^{M-1-x})_{(M-1;R),(0;R)} & \text{: $x\leq 0$, $x$ is even} \\
0 &\text{: otherwise.}
\end{cases}
\end{align*} 
\end{corollary}
%%%%%%%%%%%%%
\begin{proof}
Let the Fourier transform of $\psi \in \ell^2(\mathbb{Z},\mathbb{C}^2)$ with $\psi(x)=[\psi_L(x)\; \psi_R(x)]^{\top}$ be $\hat{\psi}(\xi)=[\psi_L(\xi)\; \psi_R(\xi)]^\top$, while let the Fourier inverse of $\hat{q}_M(\xi)$ be $\mathcal{U}^{-1}\hat{q}_M:=q_M'$ for $q\in\{a,b,c,d\}$. 
By (\ref{eq:FTS}), taking the Fourier inversion, we have  
    \begin{align*} 
    (\mathcal{U}^{-1}\hat{a}_M(\xi) \hat{\psi}_L(\xi))(x) &= (a_M'*\psi_L)(x),\;
    (\mathcal{U}^{-1}\hat{b}_M(\xi) \hat{\psi}_R(-\xi))(x) = (b_M'*J{\psi}_R)(x), \\
    (\mathcal{U}^{-1}\hat{c}_M(\xi) \hat{\psi}_L(-\xi))(x) &= (c_M'*J{\psi}_L)(x),\;
    (\mathcal{U}^{-1}\hat{d}_M(\xi) \hat{\psi}_R(\xi))(x) = (d_M'*\psi_R)(x), 
    \end{align*} 
since $f*g=\mathcal{U}^{-1}(\mathcal{U}(f)\mathcal{U}(g))$ for any $f,g\in \ell^2(Z)$. Then (\ref{eq:FTS}) implies  
    \begin{equation}\label{eq:q'} 
    \left(\Sigma_{M} \begin{bmatrix}\psi_L \\ \psi_R \end{bmatrix}\right)(y)= \left(\begin{bmatrix}a_M' & 0 \\ 0 & d_M'\end{bmatrix} * \begin{bmatrix}\psi_L \\ \psi_R \end{bmatrix}\right)(y)+\left(\begin{bmatrix}0 & b_M' \\ c_M' & 0\end{bmatrix} * \begin{bmatrix}J\psi_L \\ J\psi_R \end{bmatrix}\right)(y). 
    \end{equation}
In the following, we will see that $q_M'(x)=q_M(x)$ for $q\in \{a,b,c,d\}$. 
Consider the linear combination of $\{\Sigma_M\delta_x|L\rangle\}_{x\in \mathbb{Z}}$ such that 
    \[ (\Sigma_M\psi_L(\cdot)|L\rangle)(y)
    =\begin{bmatrix}
    \sum_{x}\psi_L(x) (\Sigma_M\delta_x|L\rangle)(y;L) \\
    \sum_{x}\psi_L(x) (\Sigma_M\delta_x|L\rangle)(y;R)
    \end{bmatrix}. \]
By Lemma~\ref{lem:scatt_M}, 
    \begin{align*}
    (\Sigma_M\delta_x|L\rangle)(y;L)
        &= \left\langle \delta_y, \sum_{j=0}^\infty a_M(2j) \delta_{x+2j} \right\rangle_{\ell^2(\mathbb{Z})}= a_M(y-x) \\
    (\Sigma_M\delta_x|L\rangle)(y;R)
        &= \left\langle \delta_y, \sum_{\ell=-(M-1)}^\infty c_M(-2\ell) \delta_{-x+2\ell} \right\rangle_{\ell^2(\mathbb{Z})}= c_M(y+x)
    \end{align*}
Then we have
    \begin{align*}
    (\Sigma_M\psi_L(\cdot)|L\rangle)(y)
    =\begin{bmatrix}
    \sum_{x}\psi_L(x) a_M(y-x) \\
    \sum_{x}\psi_L(x) c_M(y+x)
    \end{bmatrix}
    =\begin{bmatrix}
    \sum_{x}a_M(x) \psi_L(y-x) \\
    \sum_{x}c_M(x) \psi_L(-y+x) 
    \end{bmatrix}
    =\begin{bmatrix}
    (a_M*\psi_L)(y) \\
    (c_M*J\psi_L)(y) 
    \end{bmatrix}
    \end{align*}
In the same way, we obtain 
    \[ (\Sigma_M\psi_L(\cdot)|R\rangle)(y)
    = \begin{bmatrix}
    (b_M*J\psi_L)(y) \\
    (d_M*\psi_L(y) 
    \end{bmatrix}. \]
Therefore we obtain
    \[ \left(\Sigma_{M} \begin{bmatrix}\psi_L \\ \psi_R \end{bmatrix}\right)(y)= \left(\begin{bmatrix}a_M & 0 \\ 0 & d_M\end{bmatrix} * \begin{bmatrix}\psi_L \\ \psi_R \end{bmatrix}\right)(y)+\left(\begin{bmatrix}0 & b_M \\ c_M & 0\end{bmatrix} * \begin{bmatrix}J\psi_L \\ J\psi_R \end{bmatrix}\right)(y). \]
Comparing it with (\ref{eq:q'}), we obtain the desired conclusion. 
\end{proof}

Comparing the shapes of $a_M,\dots,d_M$ in Corollary~\ref{cor:abcd} with Lemma~\ref{lem:scatt_M}, we have 
\begin{align} 
a_M(x) &= (\Sigma_M \delta_{x_0}|L\rangle)_L(x_0+x),  \notag \\
b_M(x) &= (\Sigma_M \delta_{x_0}|R\rangle)_L(-x_0+x), \notag \\
c_M(x) &= (\Sigma_M \delta_{x_0}|L\rangle)_R(-x_0+2(M-1)+x),\notag \\
d_M(x) &= (\Sigma_M \delta_{x_0}|R\rangle)_R(x_0+x). \label{eq:abcd_r}
\end{align}
\begin{remark}
From (\ref{eq:abcd_r}), Propositions~\ref{PropM1}, \ref{PropM2} and \ref{PropM3} imply that  
\begin{align*}
    \cdot\; M=1:& \\
    a_1(x) &= 
    \begin{cases} 
    a & \text{: $x=0$}\\ 
    0 & \text{: otherwise}
    \end{cases},
    \;b_1(x) = 
    \begin{cases} 
    b & \text{: $x=0$}\\ 
    0 & \text{: otherwise}
    \end{cases}, \\
    c_1(x) &=
    \begin{cases}
    c & \text{: $x=0$} \\
    0 & \text{: otherwise}
    \end{cases},
    \; d_1(x)=
    \begin{cases} 
    d & \text{: $x=0$} \\
    0 & \text{: otherwise}
    \end{cases}. \\
    \\
    \cdot\; M=2:& \\
    a_2(x) &= 
    \begin{cases} 
    a^2 (bc)^{x/2} & \text{: $x\in 2\mathbb{Z}_{\geq 0}$}\\ 
    0 & \text{: otherwise} 
    \end{cases}, \;
    b_2(x)=
    \begin{cases}
    b & \text{: $x=0$}\\
    adb (bc)^{x/2-1} & \text{: $x\in 2\mathbb{Z}_{> 0}$}\\ 
    0 & \text{: otherwise} 
    \end{cases}, \\ 
    c_2(x) &=
    \begin{cases}
    c & \text{: $x=0$}\\
    acd (cb)^{-x/2-1} & \text{: $x\in 2\mathbb{Z}_{< 0}$}\\ 
    0 & \text{: otherwise} 
    \end{cases},\;
    d_2(x) =
    \begin{cases} 
    d^2 (bc)^{-x/2} & \text{: $x\in 2\mathbb{Z}_{\leq 0}$}\\ 
    0 & \text{: otherwise} 
    \end{cases}
    \end{align*}
    \begin{align*}
\cdot\;M=3:&  \\
a_3(x) 
 &= \begin{cases}
(-|c|\Delta)^{x/2}a^2\left\{ ie^{-i\gamma}b\sin[(x/2) \theta]+a\cos[(x/2) \theta] \right\} & \text{: $x\in 2\mathbb{Z}_{\geq 0}$} \\
0 & \text{: otherwise}
\end{cases},\\
b_3(x)
 &= \begin{cases} 
b & \text{: $x=0$}\\ 
(-|c|\Delta)^{x/2-1}ad\left\{ -ie^{i\gamma}a\sin[(x/2-1)\theta]+b\cos[(x/2-1)\theta] \right\} & \text{: $x\in 2\mathbb{Z}_{>0}$} \\ 
0 & \text{: otherwise} 
\end{cases}, \\
c_3(x) 
 &= \begin{cases} 
c & \text{: $x=0$}\\ 
(-|c|\Delta)^{-x/2-1}ad\left\{ ie^{-i\gamma}d\sin[(-x/2-1)\theta]+c\cos[(-x/2-1)\theta] \right\} & \text{: $x\in 2\mathbb{Z}_{< 0}$}\\
0 & \text{: otherwise} 
\end{cases}, \\
d_3(x) 
 &= \begin{cases}
(-|c|\Delta)^{-x/2-1}d^2\left\{ -ie^{i\gamma}c\sin[(-x/2) \theta]+d\cos[(-x/2) \theta] \right\} & \text{: $x\in 2\mathbb{Z}_{\leq 0}$} \\
0 & \text{: otherwise}
\end{cases}.
\end{align*}
\end{remark}
%
%%%%%%%%%%%%%%%%%%
\subsection{Relation to resonant-tunneling in discrete-time quantum walk}
%%%%%%%%%%%%%%%%%%
Let us see a meaning of the result from the view point of a quantum walk's dynamics; resonant-tunneling in discrete-time quantum walk by \cite{MMOS}. 
We set the following initial state of the quantum walk by 
 \[ \Psi_0(x)=  \begin{cases} 
 \alpha_L e^{i\xi (x-(M-1))} |L\rangle & \text{: $x\geq M-1$} \\
 \alpha_R e^{i\xi |x|} |R\rangle & \text{: $x\leq 0$} \\ 
  0 & \text{: otherwise}
\end{cases}
 \]
Here $\alpha_L$ and $\alpha_R$ are arbitraly complex values. 
Then we consider the quantum walk iterated by $\Psi_n=U\Psi_{n-1}$ with this initial state.  
The following facts have been known in more general situation including our setting. 
\begin{proposition}\cite{HS}
\begin{enumerate}
\item This quantum walk converges to a stationary state in the following meaning: 
\[ \exists \lim_{n\to \infty} e^{-i\xi n}\Psi_n(x)=:\Phi_\infty(x).\] 
\item This stationary state is a generalized eigenfunction satisfying 
\[ U\Phi_\infty=e^{i\xi}\Phi_\infty. \]
\end{enumerate}
\end{proposition}
Now let us describe $\Phi_\infty$ using $E_M$ and connect to the scattering matrix. 
Note that $\alpha_L=\Psi_0(M-1;L)=\Phi_n(M-1;L)$, $\alpha_R=\Psi_0(0;R)=\Phi_n(0;R)$ $(n=0,1,2,\dots)$ are the inputs into $\Gamma_{M}$. 
On the other hand, $\beta_L:=\Phi_\infty(-1;L)$ and $\beta_R:=\Phi_\infty(M;R)$ are the responces to the inputs in the long time limit. 
We obtain an expression of this responses using the following $2\times 2$ matrix $\tilde{S}_M(\xi)$: 
\begin{proposition}\label{propSmat}
Let the inputs of the quantum walk be $\alpha_L=\Psi_0(M-1;L)$, $\alpha_R=\Psi_0(0;R)$. 
%Let $\hat{a}_M(\xi)$, $\hat{b}_M(\xi)$, $\hat{c}_M(\xi)$, $\hat{d}_M(\xi)$ be defined in (\ref{eq:abcd}). 
The responses to the inputs in the long time limit; $\beta_L=\Phi_\infty(-1;L)$, $\beta_R=\Phi_\infty(M;R)$, are  described by 
	\[ e^{i\xi}\begin{bmatrix} \beta_L \\ \beta_R \end{bmatrix} = \widetilde{S}_M(\xi)\begin{bmatrix} \alpha_L \\ \alpha_R \end{bmatrix}, \]
where $\widetilde{S}_M(\xi)$ can be expressed by using $\widehat{a}_M(\xi),\widehat{b}_M(\xi),\widehat{c}_M(\xi),\widehat{d}_M(\xi)$ defined in (\ref{eq:abcd}) as follows: 
       \begin{equation}\label{eq:dynamicsscatt} 
        \widetilde{S}_M(\xi)=e^{-i\xi (M-1)}\begin{bmatrix} e^{i \frac{M-1}{2}\xi} & 0 \\ 0 & e^{-i \frac{M-1}{2}\xi} \end{bmatrix} 
        \begin{bmatrix} \widehat{a}_M(\xi) & \widehat{b}_M(\xi) \\ \widehat{c}_M(-\xi) & \widehat{d}_M(-\xi) \end{bmatrix} 
        \begin{bmatrix} e^{i \frac{M-1}{2}\xi} & 0 \\ 0 & e^{-i \frac{M-1}{2}\xi} \end{bmatrix}^{-1}.
       \end{equation}
\end{proposition}
\begin{proof}
Let us consider the quantum walk at time $n$. 
$\Phi_{n+1}(-1;L)$ is expressed by 
	\begin{multline*} 
        e^{i\xi}\Phi_{n+1}(-1;L) 
        = \alpha_L \big\{ae^{-i(M-1)\xi}(E_M^{M-1})_{(M-1;L)\;(0;L)}+ae^{-i(M+1)\xi} (E_M^{M+1})_{(M-1;L)\;(0;L)}\\
        	+ae^{-i(M+3)\xi} (E_M^{M+3})_{(M-1;L)\;(0;L)}+\cdots\big\} \\
          +\alpha_R \left\{b+ ae^{-2 i\xi}(E_M^{2})_{(0;L)\;(0;R)}+ae^{-4i\xi} (E_M^{4})_{(M-1;L)\;(0;L)}+\cdots\right\}.
        \end{multline*}
In the same way, $\Phi_{n+1}(M;R)$ is expressed by 
	\begin{multline*} 
        e^{i\xi}\Phi_{n+1}(M;R) 
        = \alpha_L \left\{c+ de^{-2 i\xi}(E_M^{2})_{(M-1;R)\;(M-1;L)}+de^{-4i\xi} (E_M^{4})_{(M-1;R)\;(M-1;L)}+\cdots\right\} \\
        +\alpha_R \big\{de^{-i(M-1)\xi}(E_M^{M-1})_{(M-1;R)\;(0;R)}+de^{-i(M+1)\xi} (E_M^{M+1})_{(M-1;R)\;(0;R)}\\
        	+de^{-i(M+3)\xi} (E_M^{M+3})_{(M-1;R)\;(0;R)}+\cdots\big\}.
        \end{multline*}
Taking $n\to\infty$, we obtain 
	\[ \widetilde{S}_M(\xi) = 
        	\begin{bmatrix}
ae^{-i(M-1)\xi}\left( E_M^{M-1}G_M(\xi)\right)_{(0;L),\;(M-1;L)} & b+a \left( e^{-2i\xi}E_M^2G_M(\xi)\right)_{(0;L),\;(0;R)} \\
c+d \left(e^{-2i\xi}E_M^2G_M(\xi)\right)_{(M-1;R),\;(M-1;L)} & de^{-i(M-1)\xi} (E_M^{M-1}G_M(\xi))_{(M-1;R),\;(0;R)}
                \end{bmatrix}. \]
Comparing the form of this matrix $\widetilde{S}(\xi)$ with $\widehat{a}_M(\xi),\widehat{b}_M(\xi),\widehat{c}_M(\xi),\widehat{d}_M(\xi)$ in (\ref{eq:abcd}), 
we obtain the desired expression of $\widetilde{S}(\xi)$. 
\end{proof}
Now let us see $\widetilde{S}_M(\xi)$ is a unitary matrix on $\mathbb{C}^2$ using the fact~\cite{Morioka} that 
the scattering matrix $\widehat{\Sigma}_M(\theta)$ is unitary on ${\bf h}(\theta)$. 
The scattering matrix $\widehat{\Sigma}_M(\theta)$ is expressed by 
	\[ (\widehat{\Sigma}_M(\theta) f)(\xi_s(\theta))= 
        	\begin{bmatrix} \widehat{a}_M(\xi_s(\theta)) & 0 \\ 0 & \widehat{d}_M(\xi_s(\theta)) \end{bmatrix} f(\xi_s(\theta))
               +\begin{bmatrix} 0 & \widehat{b}_M(\xi_s(\theta)) \\ \widehat{c}_M(\xi_s(\theta)) & 0 \end{bmatrix} f(-\xi_s(\theta))
                \]
for any $f(\cdot )\in {\bf h}(\theta)$, $s\in\{\pm\}$. 
Remark that the inner product of ${\bf h}(\theta)$ is for $f=f_+{\bf e}_+ + f_-{\bf e}_-= [f_+,f_-]^\top$ and 
$g=g_+{\bf e}_+ + g_-{\bf e}_-= [g_+,g_-]^\top$, 
	\[ \sum_{s\in\{\pm\}} \langle f(\xi_s(\theta)),g(\xi_s(\theta)) \rangle_{\mathbb{C}^2}. \]
The scattering matrix is unitary in ${\bf h}(\theta)$, that is, 
\[ || f ||^2_{{\bf h}(\theta)}= || \widehat{\Sigma}_M(\theta)f ||^2_{{\bf h}(\theta)}. \]
Then for $f=[f_L\;f_R]^\top\in {\bf h}(\theta)$, we have 
\begin{equation}\label{eq:unitarity} 
I_1+I_2+I_3+I_4=|f_L(\xi_+)|^2+|f_R(\xi_-)|^2+|f_L(\xi_{-})|^2+|f_R(\xi_{+})|^2, 
\end{equation}
where using the fact that $-\xi_{\pm}=\xi_{\mp}$, we put
	\begin{align*} 
        I_1:= |\widehat{a}_M(\xi_+)f_+(\xi_+)+\widehat{b}_M(\xi_+)f_R(\xi_-)|^2,\;I_2:=|\widehat{a}_M(\xi_-)f_L(\xi_-)+\widehat{b}_M(\xi_-)f_R(\xi_+)|^2, \\
        I_3:= |\widehat{c}_M(\xi_+)f_L(\xi_-)+\widehat{d}_M(\xi_+)f_R(\xi_+)|^2,\;I_4:=|\widehat{c}_M(\xi_-)f_L(\xi_+)+\widehat{d}_M(\xi_-)f_R(\xi_-)|^2.
        \end{align*}
Then (\ref{eq:unitarity}) is equivalent to for any $s\in\{\pm \}$
	\[ |\widehat{a}_M(\xi_s)|^2+|\widehat{c}_M(-\xi_{s})|^2=1, \]
        \[ |\widehat{b}_M(\xi_s)|^2+|\widehat{d}_M(-\xi_{s})|^2=1, \]
	\[ \widehat{a}_M(\xi_s)\overline{\widehat{b}_M(\xi_s)}+\widehat{c}_M(-\xi_{s})\overline{\widehat{d}_M(-\xi_{s})}=0. \]
Thus (\ref{eq:unitarity}) is equivalent to the unitarity of 
	\[ \begin{bmatrix}  \widehat{a}_M(\xi_s)  & \widehat{b}_M(\xi_s) \\ \widehat{c}_M(-\xi_s) & \widehat{d}_M(-\xi_s) \end{bmatrix} \]
in $\mathbb{C}^2$. 
Thus returning back to the expression of $\widetilde{S}_M(\xi)$, (\ref{eq:dynamicsscatt}), 
since 
	\[ \begin{bmatrix} e^{i(M-1)\xi} & 0 \\ 0 & 1 \end{bmatrix}\; \mathrm{and} \; \begin{bmatrix} 1 & 0 \\ 0 & e^{i(M-1)\xi} \end{bmatrix} \]
are unitary matrices, then $\widetilde{S}(\xi)$ is also a unitary matrix. 

Finally, using this fact of the unitary, let us see, for example, the case for $M=2$ case. 
From Proposition~\ref{propSmat} and Remark~\ref{remabcd}, 
	\[ \widetilde{S}_2(\xi)=
        \frac{1}{1-e^{-2i\xi}bc}
        \begin{bmatrix} 
        a^2e^{-i\xi} & (1+\Delta e^{-2i\xi}) b \\
        (1+\Delta e^{-2i\xi}) c & d^2e^{-i\xi} 
        \end{bmatrix}.\]
Then the unitarity of $\widetilde{S}_2(\xi)$ implies that 
the perfect transmitting happens iff $1+\Delta e^{-2i\xi}=0$. 
This condition agrees with the result on \cite{MMOS}. 

%The Fourier transform of an initial state so that the origin $x=0$ obtains a constant ``out source" from the negative side at every time step 
%is expressed by $\phi_+(\xi)=\sum_{x\leq 0}e^{i\xi x}|R\rangle$. 
%Putting $\phi_-(\xi):=\sum_{x\geq 0}e^{i\xi x}|L\rangle$, we have 
%	\[ \hat{\Sigma}_M(\xi)\phi_+(\xi)=\hat{d}_M(\xi)\phi_+(\xi)+\hat{b}_M(\xi)\phi_-(\xi). \]
%Then the reflection and transmitting values for the energy $\xi$ can be expressed by 
%	\[ R_M(\xi):=|\hat{b}_M(\xi)|^2,\;T_M(\xi):=|\hat{d}_M(\xi)|^2. \]
%For example, if $M=2$, then the condition $1+\Delta e^{2i\xi}=0$ gives the perfectly transmitting, that is, $R_2(\xi)=0$ and $T_2(\xi)=1$, 
%which coincides with \cite{MMOS}. 
%\section{Summary and discussions}

\noindent\\
\noindent {\bf Acknowledgments}
H. M. was supported by the grant-in-aid for young scientists No.~16K17630, JSPS. 
E.S. acknowledges financial supports from the Grant-in-Aid of
Scientific Research (C) Japan Society for the Promotion of Science (Grant No.~19K03616) and Research Origin for Dressed Photon.

%\appendix
%\def\thesection{Appendix \Alph{section}}
%\renewcommand{\theequation}{A.\arabic{equation}}
%\setcounter{equation}{0}

%\section{}

\begin{small}
\bibliographystyle{jplain}

\end{small}

\end{document}